\documentclass[letterpaper]{article}
\usepackage{aaai2027}
\nocopyright

\usepackage[hyphens]{url} 
\usepackage{graphicx} 
\usepackage{natbib} 
\usepackage{caption} 
\usepackage[T1]{fontenc}
\usepackage{graphicx}
\usepackage{amsthm}
\usepackage{amsmath,amssymb,amsfonts}
\usepackage{bm}
\usepackage{stmaryrd}
\usepackage{xspace}
\usepackage{cancel}
\usepackage[T1]{fontenc}
\usepackage{multirow}%
\usepackage{mathrsfs}%
\usepackage[title]{appendix}%
\usepackage{textcomp}%
\usepackage{manyfoot}%
\usepackage{booktabs}%
\usepackage{algorithm}%
\usepackage{algpseudocode}%
\usepackage{listings}%
\usepackage{cite}
\usepackage{lineno}
\usepackage{array}
\usepackage{amsmath,amssymb,amsfonts}
\usepackage{graphicx}
\usepackage{textcomp}
\usepackage{xcolor}
\usepackage{bm}
\usepackage{stmaryrd}
\usepackage{xspace}
\usepackage{tikz}
\usepackage{xcolor}
\usepackage{pifont}
\usetikzlibrary{positioning}
\usepackage{algorithm}
\usepackage{algpseudocode}

\usepackage[most]{tcolorbox}
\usepackage{xcolor}
\usepackage{amsmath, amssymb}

\definecolor{problemblue}{RGB}{232,244,255}
\definecolor{problemframe}{RGB}{60,120,180}

\newcommand{\shortneg}{\scalebox{0.6}[1]{$-$}}

\newtcolorbox{problembox}[1][]{
    colback=problemblue,
    colframe=problemframe,
    coltitle=black,
    title=#1,
    boxrule=0.6pt,
    arc=2pt,
    left=4pt,
    right=4pt,
    top=4pt,
    bottom=4pt,
    width=\linewidth,
    before skip=8pt,
    after skip=8pt,
    breakable
}

\theoremstyle{definition} 
\newtheorem{definition}{Definition}

\newtheorem{proposition}{Proposition}

\theoremstyle{problem}
\newtheorem{subproblem}{Subproblem}

\theoremstyle{remark}

\title{$S^3$: A \underline{S}mooth \underline{S}imulation \underline{S}urrogate for \\ Optimizing Discrete Abstractions of Dynamical Systems}

\author{
    Jordan Peper,
    James Mathias Gast,
    Vignesh Nanduri, \\
    Tanmayee Maram,
    Ethan Howes,
    Shivani Chandrasekar,
    Ivan Ruchkin
    \\
}
\affiliations{
    University of Florida\\

    \{jpeper, jgast, vignesh.nanduri, tanmayeemaram, ethan.howes, shivani.chandras\}@ufl.edu, iruchkin@ece.ufl.edu
}

\begin{document}

\maketitle

\begin{abstract}
Intelligent systems are increasingly deployed in safety-critical settings with black-box controllers, including neural networks. The properties and behaviors of these end-to-end systems can be studied with abstraction-based methods that replace them with simpler finite models. Constructing such abstractions requires balancing the \emph{soundness} of over-approximating the dynamical system against \emph{conservatism}, which manifests as spurious or excessive nondeterministic behaviors. Bi-simulation theory provides principled metrics for characterizing these relationships, but does not prescribe how to construct sound abstractions with minimal conservatism. We fill this gap with a \emph{smooth simulation surrogate} ($S^3$) --- a differentiable objective that approximates the \emph{reverse simulation metric} used to quantify conservatism. Combined with Taylor model-based reachability, $S^3$ enables gradient-based optimization of abstraction parameters while preserving soundness by construction. We evaluate this optimization pipeline on three case studies. Our results show that $S^3$ is strongly correlated with the reverse simulation metric, is computationally faster, and serves as an effective objective for reducing abstraction conservatism. \textbf{Code:} https://github.com/Trustworthy-Engineered-Autonomy-Lab/s3-abstraction-optimization


\end{abstract}

\section{Introduction}

\looseness=-1
Abstraction-based methods are widely used to study complex dynamical systems indirectly through simpler abstract models~\cite{baier_principles_2008, alur_formal_2011, kloetzer_fully_2008, belta_formal_2017, tabuada_verification_2009}. These methods reduce a potentially infinite-state dynamical system to a finite-state representation by (1) partitioning the continuous state space into a collection of \emph{abstract states} and (2) constructing \emph{transitions} between those states that capture the behavior of the underlying dynamical system. When constructed appropriately, the abstract model will capture the reachable sets of the concrete system. This enables tractable yet sound analysis of semantically rich properties, such as whether a system repeatedly visits locations in a prescribed order, or synthesizing a controller that provably guarantees it.

It is often essential that an abstract model \emph{over-approximates} the dynamical system (referred to as \emph{soundness}), particularly in safety-critical applications. However, enforcing soundness can also introduce excessive conservatism into the abstraction. Overly conservative abstractions may exhibit substantial nondeterminism and spurious transitions not realizable by the underlying dynamical system. As a result, analysis of such models may yield pessimistic conclusions: properties that hold for the dynamical system may nevertheless be impossible to certify using the abstraction. Increasing the granularity of the state-space partition can reduce conservatism, but doing so also enlarges the abstract model and increases the computational costs~\cite{soudjani_adaptive_nodate, fisher_approximating_2008, clarke_grumberg_jha_lu_veith_2000, clarke_lu_grumberg_jha_veith_2003, clarke_fehenker_han_krogh_2003}. Hence, available computational resources impose a practical limit on the number of abstract states and transitions. Thus, a long-standing design question is: How should finite abstract states be distributed across the state space to reduce conservatism while preserving soundness?

Typical methods balance soundness and conservatism in the context of an unknown, to-be-synthesized controller using uniform rectilinear partitioning of the continuous state space. 
Tools such as SCOTS~\shortcite{runggerSCOTSToolSynthesis2016} over-approximate reachable sets via a growth bound, while PESSOA~\shortcite{mazo_pessoa_2010} and CoSyMA~\shortcite{mouelhi_cosyma_2013} govern bi-simulation error primarily through partition granularity. Other approaches optimize cell geometry, for example by selecting aspect ratios that reduce average transition count under a fixed cell volume~\cite{weber_optimized_2017, ajeleye_data-driven_2023}. 
To our knowledge, no existing method exploits a known controller to directly optimize conservatism while adhering to soundness and a fixed computation budget. 

This paper aims to reduce the conservatism of abstract models of closed-loop dynamical systems, while preserving soundness under a limited computational budget, by optimizing the \emph{parameters} of the state quantization. We ground our investigation in a principled framework for characterizing soundness and conservatism --- \emph{(bi)simulation theory}~\cite{park_concurrency_1981, milner1989communication}. One system \emph{simulates} another if each transition of the former can be matched by a corresponding transition of the latter (i.e., a \emph{simulation relation} exists). When exact simulation cannot be established, as is often the case when one system is discrete and the other is continuous, one system \emph{approximately simulates} another if the distance between their states, or outputs, remains within a prescribed bound before and after each simulation step~\cite{girardApproximateBisimulationBridge2011, pola_approximately_2008}. This bound, called the \emph{simulation metric}, quantifies either an abstract model's \emph{soundness} (forward simulation) or its \emph{over-conservatism} (reverse simulation).

\looseness=-1
Unfortunately, these simulation metrics are poorly suited as optimization objectives: they are computationally expensive to evaluate, non-differentiable with respect to the quantization parameters, and characterize only the worst-case discrepancy (rather than the aggregate). To overcome these limitations, we derive a computationally efficient and \emph{smooth simulation surrogate} ($S^3$), which quantifies over-conservatism. Starting from the Bellman equation associated with the reverse simulation metric, we construct a strict upper-bounding value function and a smooth finite-horizon approximation that closely predicts the metric. Combined with sound-by-construction transitions based on Taylor models and bounded remainders~\cite{chen2015reachability}, this proxy enables gradient-based optimization of the abstraction's quantization parameters. To our knowledge, this is the first differentiable proxy of an upper-bound on the reverse simulation metric for a discrete abstractions of a dynamical system.

The contributions of this paper are threefold:
\begin{enumerate}
    \item A differentiable and computationally efficient proxy measure of abstraction conservatism.
    \item A gradient-based optimization pipeline to construct minimally conservative abstractions by minimizing this proxy.
    \item Case studies showing strong empirical correlation and optimization utility to improve the real simulation metric.
\end{enumerate}

\subsubsection{Related Work.} Abstraction-based methods are widely used to study properties of complex dynamical systems indirectly through simpler abstract models. Typical state abstraction schemes include \textit{phase-based}~\cite{frehse_phaver_2008, frehse_spaceex_2011, sloth_algorithmic_2011}, where rectangular or linear flow constraints are placed on the system states and dynamics, and \textit{predicate-based}~\cite{alur_counterexample-guided_2006, alur_predicate_2006, kloetzer_fully_2008}, where continuous state spaces are abstracted as predicate-defined subsets. State-of-the-art tools for abstracting nonlinear system dynamics include $\mathrm{Flow}^*$~\cite{hutchison_flow_2013} which leverages Taylor models to over-approximate reachable sets. When dynamics are not known, \cite{coppola_data-driven_2023, devonportSymbolicAbstractionsData2021a, lavaeiDataDrivenSynthesisSymbolic2022} sample trajectories from the system and synthesize probabilistic guarantees on model correctness.

Furthermore, a growing body of work incorporates optimization into the abstraction pipeline. One line of research focuses on \emph{state aggregation}, learning partitions of the state space that preserve the Markov property and improve sampling efficiency, though typically without soundness guarantees \cite{duan_state_2019, allenLearningMarkovState2024}. For systems with bounded disturbances, \cite{weber_optimized_2017, ajeleye_data-driven_2023} find rectilinear aspect ratios that minimize average transition count. Several methods also learn the parameters of the transition system itself or alternating bisimulation functions via robust optimization \cite{lavaeiDataDrivenSynthesisSymbolic2022}) and may be paired with PAC-style guarantees. Finally, for verification tasks, counterexample-guided abstraction refines spurious trajectories that falsely violate a specification~\cite{dierks_automatic_2007, goos_verification_2003, clarke_abstraction_2003, fehnker_refining_2005, alur_counterexample-guided_2006}.

The rest of this paper is organized as follows. First, we introduce preliminaries and state the problem. Then, we guarantee forward simulation via conservative reachability. Next, we prove an upper bound on the simulation Bellman function, and derive a smooth approximation of this upper bound. We discuss the end-to-end optimization pipeline, and then deploy this pipeline to three case studies and compare it to two baseline strategies for reducing conservatism. 

\section{Preliminaries and Problem}
\label{sec:preliminary}


Let $\mathbb{R}^n$ and $\mathbb{N}^n$ denote the $n$-dimensional set of real numbers and natural numbers (excluding 0), respectively. For sets $A$ and $B$, notation $A \cap B$ and $A \cup B$ denote intersection and union, respectively. $A \subset B$ and $A \subseteq B$ indicate $A$ is a proper and improper subset of $B$, respectively. The Cartesian product between sets $A$ and $B$ is denoted $A \times B$. The Minkowski sum between sets $A$ and $B$ is written $A \oplus B = \{a + b \mid a \in A, b \in B\}$. We denote the powerset of $A$ by $2^A$. Symbols $\forall$ and $\exists$ denote the ``for all'' and ``exists'' quantifiers. An $n$-dimensional axis-aligned bounding box (AABB) is defined by $[\boldsymbol{\ell}, \boldsymbol{u}]$, where $\boldsymbol{\ell} \in \mathbb{R}^n$ is the minimum corner and $\boldsymbol{u} \in \mathbb{R}^n$ is the maximum corner.

\subsection{Transition systems}

We introduce several relevant definitions necessary to understand the problem at hand. First, we assume that a discrete-time dynamical system is an infinite-state transition system.

\begin{definition} [Infinite-state transition system]
\label{def:ists}
    An infinite-state transition system is a tuple $s = (X, X_0, f)$, where:
    \begin{itemize}
        \item $X \subset \mathbb{R}^n$ is an $n$-dimensional state space
        \item $X_0 \subseteq X$ is a subset of initial states
        \item $f : X \rightarrow X$ is the dynamical transition function
    \end{itemize}
\end{definition}

We assume that the dynamical equations of motion $f$ are known, deterministic, and at least twice-differentiable on $X$. We denote the image of $A \subseteq X$ under $f$ as $f(A) = \{f(x) \mid x \in A\}$. Furthermore, $f^k$ denotes the $k$-fold iterate of $f$.

\begin{definition} [Finite-state transition system]
\label{def:fsts}
    A finite-state transition system is a tuple $\hat{s} = (\hat{X}, \hat{X}_0, \hat{f})$, where:
    \begin{itemize}
        \item $\hat{X} \subset \mathbb{N}^n$ is an $n$-dimensional abstract state set
        \item $\hat{X}_0 \subseteq \hat{X}$ is a subset of initial abstract states
        \item $\hat{f} : \hat{X}  \rightarrow 2^{\hat{X}}$ is the transition function
    \end{itemize}
\end{definition}

This system is \emph{nondeterministic}: each state $\hat{x}$ may admit multiple valid successors $\hat{f}(\hat{x}) \subseteq \hat{X}$. We denote the image of $\hat{A} \subseteq \hat{X}$ under $\hat{f}$ as $\hat{f}(\hat{A}) = \{\hat{f}(\hat{x}) \mid \hat{x} \in \hat{A}\}$. Furthermore, $\hat{f}^k$ denotes the $k$-fold iterate of $\hat{f}$.


\subsection{Sound abstractions}

Abstract modeling is a process whereby an infinite-state transition (concrete) system is modeled as a finite-state transition (abstract) system whose states represent subsets of the concrete state space and transitions reflect those of the concrete dynamics. State abstraction is performed by a \emph{state quantization} function $\psi_\theta : X \rightarrow \hat{X}$, where $\theta$ denotes the parameters of the quantizer from a parameter space $\Theta$. For any $\hat{x} \in \hat{X}$, define the quantization cell (preimage) as $\Psi_\theta(\hat{x})=\{x\in X \mid \psi_\theta(x) = \hat{x} \}$.

After the state space is abstracted into a finite collection of abstract states, transitions are built over the abstract state set. If the resulting finite-state transition system contains all concrete states, initial states, and one-step transitions under this quantization, we call it a \emph{sound abstract model}.

\begin{definition} [Sound abstract model]
\label{def:cons-abstraction}
    Let $s$ be an infinite-state transition system (Def.~\ref{def:ists}) and let $\hat{s}_\theta$ be a finite-state transition system (Def.~\ref{def:fsts}) related to $s$ through the a state quantization function $\psi_\theta$. The system $\hat{s}$ is a \emph{sound abstraction} of $s$ if the following conditions hold:
    \begin{enumerate}
        \item $X \subseteq \bigcup_{\hat{x} \in \hat{X}} \Psi_\theta(\hat{x})$ and $X_0 \subseteq \bigcup_{\hat{x}_0 \in \hat{X}_0} \Psi_\theta(\hat{x}_0)$
        \item $\forall \hat{x} \in \hat{X},f(\Psi_\theta(\hat{x})) \subseteq \bigcup_{\hat{x} \in \hat{f}(\hat{x})} \Psi_\theta(\hat{x})$
    \end{enumerate}
\end{definition}

Condition 1 ensures that every concrete state and initial state have one conservative abstract representative. Condition 2 ensures that every one-step transition of the concrete system is represented by a transition in the abstract model.

\subsection{Approximate simulation relations}

A classical way to study the relationship between two models is through approximate simulation relations. Define a distance metric between any concrete and abstract state:
\begin{equation}
    d(x, \hat{x}) = \inf_{z \in \Psi_\theta(\hat{x})} \|x-z\|_p, \quad p \ge 1
\end{equation}
Then, we define an approximate simulation relation of $s$ by $\hat{s}$, namely, a \emph{forward $\varepsilon$-approximate simulation relation}.

\begin{definition} [Forward $\varepsilon$-Approximate Simulation Relation]
\label{def:forward-sim}
    $R_\varepsilon \subseteq X \times \hat{X}$ is an $\varepsilon$-approximate simulation relation of $s$ by $\hat{s}$ if, for all $(x, \hat{x}) \in R_\varepsilon$:
    \begin{enumerate}
        \item $d(x, \hat{x}) \le \varepsilon$
        \item $\exists \hat{x}' \in \hat{f}(\hat{x}) : (f(x), \hat{x}') \in R_\varepsilon$
    \end{enumerate}
\end{definition}

If such a simulation relation exists, then $\hat{s}$ approximately simulates $s$ with at least precision $\varepsilon$, denoted $s \lesssim_{\varepsilon} \hat{s}$. When $\varepsilon = 0$, the simulation of $s$ by $\hat{s}$ is \emph{exact}, and $\hat{s}$ is sound with respect to $s$, denoted $s \lesssim \hat{s}$. Furthermore, we define the \emph{forward simulation metric} as:
\begin{equation}
    \sigma^{\rightarrow}(s, \hat{s}) = \inf\{ \varepsilon \ge 0 : \exists R_\varepsilon \;\text{satisfying Def~\ref{def:forward-sim}} \}.
\end{equation}
Conversely, we can define an approximate simulation of $\hat{s}$ by $s$, namely, a \emph{reverse $\delta$-approximate simulation relation}.



\begin{definition} [Reverse $\delta$-Approximate Simulation relation]
\label{def:reverse-sim}
    $R_\delta \subseteq X \times \hat{X}$ is a $\delta$-approximate simulation relation of $\hat{s}$ by $s$ if, for all $(x, \hat{x}) \in R_\delta$:
    \begin{enumerate}
        \item $d(\hat{x}, x) \le \delta$
        \item $\forall \hat{x}' \in \hat{f}(\hat{x}), (f(x), \hat{x}') \in R_\delta$

    \end{enumerate}
\end{definition}

We denote the approximate simulation of $\hat{s}$ by $s$ as $s \gtrsim_\delta \hat{s}$, and the exact simulation as $s \gtrsim \hat{s}$ when $\delta = 0$. Furthermore, the \emph{reverse simulation metric} is given by:
\begin{equation}
\label{eqn:reverse-sim-metric}
    \sigma^{\leftarrow}(s, \hat{s}) = \inf\{ \delta \ge 0 : \exists R_\delta \;\text{satisfying Def.~\ref{def:reverse-sim}} \}.
\end{equation}


\subsection{Problem statement}
\label{sec:problem}

Let $s = (X, X_0, f)$ be a dynamical system, and let $\psi_\theta$ be a parameterized quantization of $X$ that induces the abstract model $\hat{s}_\theta=(\hat{X},\hat{X}_0,\hat{f})$. The central problem in this paper is to minimize the reverse simulation metric subject to exact simulation of $s$ by $\hat{s}$ by optimizing $\theta$. Formally:
\begin{equation}
    \min_{\theta \in \Theta} \sigma^{\leftarrow}(s, \hat{s}_\theta) \quad \mathrm{s.t.} \quad \sigma^{\rightarrow}(s, \hat{s}_\theta) = 0.
\end{equation}
However, solving this is generally unattainable or computationally intractable for several reasons. First, the evaluation of $\sigma^{\leftarrow}$ and $\sigma^{\rightarrow}$ is computationally expensive, requiring both the compilation of $s$ and $\psi_\theta$ into $\hat{s}_\theta$ and a sub-optimization loop to tightly approximate the simulation metrics. Second, $\sigma^{\leftarrow}$ and $\sigma^{\rightarrow}$ are generally non-differentiable in $\Theta$, requiring sample-heavy gradient-free optimization strategies.

To this end, we devise three subproblems. First, we aim to devise an abstraction-building pipeline that certifies exact simulation of $s$ by $\hat{s}$ by construction:

\begin{subproblem}
    Given $s$, choose a model parameter space $\Theta$, a quantization function $\psi_\theta$, and an algorithm to construct $\hat{s}_\theta$ that ensures $s \lesssim \hat{s}_\theta$ for all $\theta \in \Theta$.
\end{subproblem}

Next, we aim to derive a differentiable proxy of the reverse simulation metric that is computable directly from the parameters and concrete system:

\begin{subproblem}
    Obtain a proxy measure $\tilde{\sigma}^{\leftarrow}(s, \theta) \approx \sigma^{\leftarrow}(s, \hat{s}_\theta)$ that is differentiable in $\Theta$.
\end{subproblem}

Finally, we aim to solve the optimization problem:
\begin{subproblem}
    Given a parameter search space $\Theta$ and a cost function $\tilde{\sigma}^{\leftarrow}(s, \theta)$, solve the optimization problem:
\begin{equation*}
    \min_{\theta \in \Theta} \tilde\sigma^{\leftarrow}(s, \theta) \quad \mathrm{s.t.} \quad \sigma^{\rightarrow}(s, \hat{s}_\theta) = 0
\end{equation*}
\end{subproblem}

\section{Exact Forward Simulation by Construction}
\label{sec:soundness}

In this section, we address the first subproblem of choosing a parameter space, a quantization function, and a model-building algorithm to guarantee exact forward simulation ($s \lesssim \hat{s}_\theta$). First, we formalize a rectilinear grid state quantization scheme. Then, following standard Taylor model-based reachability, we show that both forward simulation and soundness can be guaranteed by construction.


\subsection{Rectilinear state quantization}

Let $X = \prod_{i=1}^n[\underline{x}_i, \overline{x}_i]$ be a rectangular domain. For each dimension $i = 1, \dots, n$, we allocate a grid resolution $m_i$, such that $\hat{X} = \prod_{i=1}^n \{1, \dots, m_i \}$. Furthermore, we denote the quantization parameters as $\theta = \{\omega_i\}_{i=1}^n$, where $\omega_i \in \mathbb{R}^{m_i}$ are the gap weights. The exact interval gaps $\eta$ between neighboring hyperplanes are recovered from these weights through activation $\mathrm{Softplus}(\omega) = \ln(1+e^\omega)$, and then normalized so that they sum to the domain side length:
\begin{equation}
    \eta_{i,j} = (\overline{x}_i - \underline{x}_i)\frac{\mathrm{Softplus}(\omega_{i,j})}{\sum_{l=1}^{m_i} \mathrm{Softplus}(\omega_{i,l})}.
\end{equation}
The choice of $\mathrm{Softplus}$ maps unconstrained weights to positive interval widths and preserves their ordering.

The explicit quantization cell boundaries, illustrated in Figure~\ref{fig:quantization}, are defined by the cumulative sum over these gaps:
\begin{equation*}
    b_{i,0} = \underline{x}_i, \quad  b_{i,j} = \underline{x}_i + \sum_{l=1}^j \eta_{i,l}, \quad j = 1,\dots m_i.
\end{equation*}
Hence, the quantization of any state $x \in X$ is given by:
\begin{equation}
\label{eqn:explicit-quantization}
    \psi_\theta(x) = (j_1, \dots, j_n), \quad j_i = \min\{ j \mid x_i \le b_{i, j}\},
\end{equation}
and, the concretization of any $\hat{x} = (\hat{x}_1, \dots, \hat{x}_n)$ is given by:
\begin{equation}
\label{eqn:explicit-concretization}
    \Psi_\theta(\hat{x}) = \prod_{i=1}^n[b_{i,\hat{x}_i}, b_{i,\hat{x}_i+1}].
\end{equation}
This parameterization transforms a constrained problem of placing ordered boundaries within $[\underline{x}_i,\overline{x}_i]$ into an unconstrained one over $\Theta = \mathbb{R}^{\sum m_i}$. See the appendix for the quantization algorithm. 

\begin{figure}
    \centering
    \includegraphics[width=0.9\linewidth]{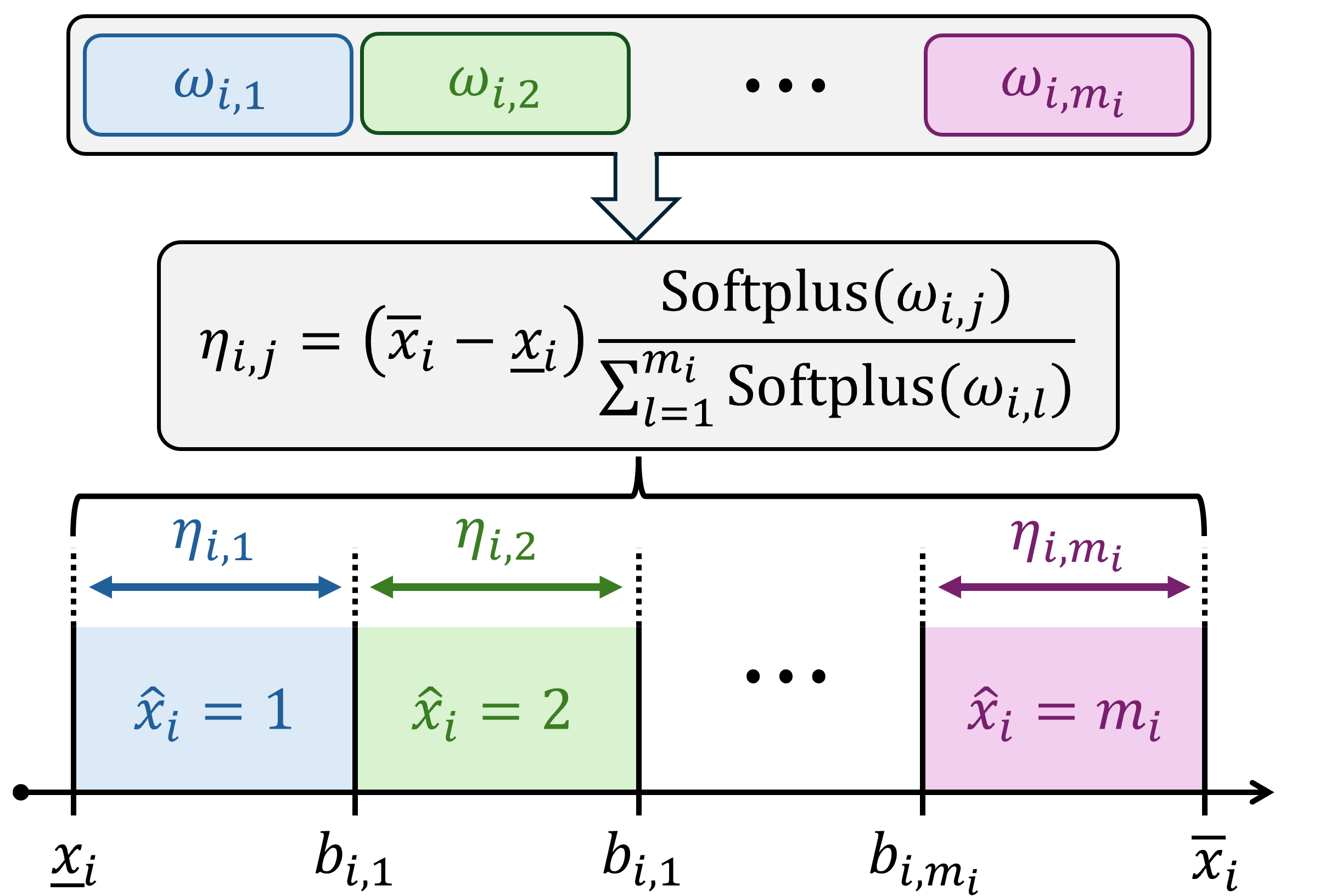}
    \caption{The quantization process forming variable-width cells from unconstrained parameters $\omega_i$ in $\theta$.}
    \label{fig:quantization}
\end{figure}

\subsection{Identifying successors through reachability}

Our goal is to synthesize, for each $\hat{x} \in \hat{X}$, a set of successors $\hat{f}(\hat{x})$ satisfying the ``behavior inclusion'' condition in Definition~\ref{def:cons-abstraction}:
\begin{equation*}
    \forall \hat{x} \in \hat{X},f(\Psi_\theta(\hat{x})) \subseteq \bigcup_{\hat{x} \in \hat{f}(\hat{x})} \Psi_\theta(\hat{x}).
\end{equation*}
To this end, we leverage first-order Taylor models --- commonly used in reachability analysis~\cite{chen2015reachability} --- and a known upper-bound on approximation error to over-approximate the one-step image of any axis-aligned bounding box (AABB) in $X$ under $f$. Then, we will utilize this formulation to make explicit the set of successors $\hat{f}(\hat{x})$.

Let $A=[\boldsymbol{\ell},\boldsymbol{u}]$ be any AABB in $X$, where $\boldsymbol{\ell}(\hat{x})$ and $\boldsymbol{u}(\hat{x})$ are its lower and upper corners. The centroid of $A$ is $c =(\boldsymbol{\ell} + \boldsymbol{u})/2$. The first-order Taylor model of $f$ at $c$ is
\begin{equation*}
F(x; c)=f(c)+\left.\frac{\partial f}{\partial x}\right|_{x=c}(x-c).
\end{equation*}
For each dimension $i=1,\dots,n$, the incurred, component-wise linearization error $e_i(x)=f_i(x)-F_i(x; c)$ is bounded by the corresponding Taylor remainder:
\begin{equation*}
|e_i(x)| \le \frac{1}{2} \max_{z \in A}\left\| \left.\frac{\partial^2f}{\partial x_i \partial x}\right|_{x=z} \right\| \| x - c\|^2.
\end{equation*}
Collecting these component-wise bounds gives the \textit{remainder vector} $\boldsymbol{e}(x) = [|e_1(x)|, \dots, |e_n(x)|]$. The resulting one-step reachable-set over-approximation of $f(A)$ is
\begin{equation}
\label{eqn:reachable-set}
\small
    \mathrm{Reach}(A) = \prod_{i=1}^n \left[  \inf_{z \in A } F_i(z;c ), \sup_{z \in A }F_i(z;c)  \right] \oplus [\shortneg \overline{\boldsymbol{e}}, \overline{\boldsymbol{e}}].
\end{equation}
The left-hand side of $\oplus$ is the AABB of $A$ under $F$, which can be determined explicitly since $F$ is linear. The remainder vector $\boldsymbol{e}(x)$ is maximized within $A$ at either of the corners, hence $\overline{\boldsymbol{e}} = \boldsymbol{e}(\boldsymbol{u}) = \boldsymbol{e}(\boldsymbol{\ell})$ denotes the worst-case approximation error. This guarantees that $f(A) \subseteq \mathrm{Reach}(A)$.

Using this reachable set over-approximation operator, we define the successors of any $\hat{x} \in \hat{X}$ as:
\begin{equation}
\label{eqn:post-set}
\hat{f}(\hat{x})=\left\{\hat{x}'\in\hat{X}\mid \Psi_\theta(\hat{x}')\cap\mathrm{Reach}(\Psi_\theta(\hat{x}))\neq\emptyset\right\},
\end{equation}
satisfying the behavior inclusion condition.

\begin{proposition}[Exact simulation by construction]
\label{prop:exact-simulation}
    Consider a dynamical system $s = (X, X_0, f)$, a parameter vector $\theta \in \Theta$, and a rectilinear state quantization function $\psi_\theta(x)$. Let the abstract state set be $\hat{X} = \psi_\theta(X)$, the initial abstract state set be $\hat{X}_0 = \psi_\theta(X_0)$, and let $\hat{f}(\hat{x})$ satisfy Equation~\ref{eqn:post-set} for all $\hat{x} \in \hat{X}$, constituting $\hat{s}_\theta = (\hat{X}, \hat{X}_0, \hat{f})$. Then, it follows that $\hat{s}_\theta$ simulates $s$ exactly, that is, $s \lesssim \hat{s}_\theta$.
\end{proposition}

\noindent
See the appendix for the proof and model-building algorithm.

\section{Smooth Simulation Surrogate}
\label{sec:proxy}

In this section, we address Subproblem 2 and 3 by deriving a proxy measure of the reverse simulation metric $\sigma^\leftarrow$ --- namely, the smooth simulation surrogate ($S^3$) --- and then optimize the quantization parameters with this objective function. First, we review the Bellman-like simulation functions, and determine computable exact finite-horizon analogues. We derive an abstract value function that is used to compute an upper bound on the simulation metric on a finite time horizon and, ultimately, $S^3$. We wrap up by discussing the details of our optimization pipeline.

\subsection{Finite-horizon simulation functions}

Girard and Pappas~\shortcite{girardApproximateBisimulationBridge2011} introduce a Bellman-like \emph{simulation function} $V$ for computing the simulation metric, measuring the largest discrepancy witnessed along the trajectories of two systems. We adapt this equation to the present setting, where $s$ is deterministic:
\begin{align*}
    V(x, \hat{x}) &= \max \left\{ d(x, \hat{x}), \sup_{\hat{x}' \in \hat{f}(\hat{x})} V(f(x), \hat{x}')\right\} \\
    & = \sup_{k \ge 0,\;\hat{x}_{k+1} \in \hat{f}(\hat{x}_k)} d(f^k(x), \hat{x}_k),
\end{align*}
where $\hat{x}_0 = \hat{x}$. The reverse simulation metric of $\hat{s}$ by $s$ can be computed from this simulation function:
\begin{equation}
    \sigma^\leftarrow(s, \hat{s}) = \sup_{\hat{x} \in \hat{X}} \inf_{x \in X} V(x, \hat{x})
\end{equation}
However, evaluating $V$ over an infinite horizon is difficult. Lyapunov-like simulation functions can provide conservative bounds on $V$, and therefore on $\sigma^\leftarrow$, but constructing them may be challenging, especially for nonlinear systems~\cite{girardApproximateBisimulationBridge2011, pola_approximately_2008}.

For abstraction design, the discrepancies often relevant to a task often occur over a finite temporal window. This is particularly true for finite-time specifications, including co-safety tasks, whose satisfaction can be determined by a finite trajectory prefix. Therefore, a finite-horizon simulation function provides a measure of over-conservatism while avoiding the full infinite-horizon evaluation. We define these finite-horizon $H$ analogues of $V$ and $\sigma^\leftarrow(s, \hat{s})$:
\begin{equation}
\label{eqn:finite-bellman} 
    V_H(x, \hat{x}) = \sup_{0 \le k \le H, \,\,\hat{x}_{k+1} \in \hat{f}(\hat{x}_k)}d(f^k(x), \hat{x}_k),
\end{equation}
\begin{equation}
\label{eqn:simulation-from-bellman}
    \sigma_H^\leftarrow(s, \hat{s}) = \sup_{\hat{x} \in \hat{X}} \inf_{x \in X} V_H(x, \hat{x}).
\end{equation}
Since the set of admissible rollout prefixes expands monotonically with $H$, $V_H(x,\hat{x})$ is nondecreasing in $H$ and satisfies $\lim_{H \rightarrow \infty} V_H(x,\hat{x}) = V(x,\hat{x})$. Consequently, it also follows that $\lim_{H \rightarrow \infty} \sigma_H^\leftarrow(s,\hat{s}) = \sigma^\leftarrow(s,\hat{s})$.


\subsection{Smooth simulation surrogate ($S^3$)}

Equation~\ref{eqn:simulation-from-bellman} quantifies the worst-case finite-horizon discrepancy over initial state pairs $(x_0,\hat{x}_0)$, but is non-differentiable and does not provide an abstraction-wide objective for optimizing the quantization parameters. We therefore proceed in two stages. First, we derive an \emph{abstract value function} that assigns a discrepancy value to each abstract state and provably upper-bounds $\sigma_H^\leftarrow$. We then replace its non-differentiable reachable sets with smooth over-approximations and smooth the remaining suprema to obtain $S^3$.

For a fixed initial abstract state $\hat{x}_0$, denote its $k^{\text{th}}$ reachable set by $\hat{A}_k = \hat{f}^k(\hat{x}_0)$. Using the $\mathrm{Reach}$ operator, this recursion can be written analogously to Equation~\ref{eqn:post-set}:
\begin{equation}
\label{eqn:abs-reach-set}
    \hat{A}_{k+1} = \{\hat{x} \in \hat{X} \mid \Psi_\theta(\hat{x}) \cap\mathrm{Reach}(\Psi_\theta(\hat A_k)) \ne \emptyset \},
\end{equation}
where $\hat{A}_0=\{\hat{x}_0\}$, $k \ge 0$. Let $c_0$ be the centroid of $\Psi_\theta(\hat{A}_0)$, and let $\hat{c}_k$ and $\hat{r}_k$ denote the centroid and radius of $\Psi_\theta(\hat{A}_k)$, respectively. We define the \emph{abstract value function} by
\begin{equation}
\label{eqn:abs-value-fcn}
    \hat{W}_H(\hat{x}) = \sup_{0 \le k \le H} \{ \hat{r}_k + |f^k(c_0) - \hat{c}_k|\},
\end{equation}
which upper-bounds $\inf_{x \in X} V_H(x, \hat{x})$ in Equation~\ref{eqn:simulation-from-bellman}.

\begin{proposition}[Upper bound on $\sigma_H^\leftarrow$]
\label{prop:upper-bound}
Given $\hat{W}_H(\hat{x})$ per Equation~\ref{eqn:abs-value-fcn}, it follows that $\sigma_H^\leftarrow \le \sup_{\hat{x} \in \hat{X}} \hat{W}_H(\hat{x})$.
\end{proposition}

\noindent
See the appendix for the proof. This upper bound remains non-differentiable for two reasons. First, the construction of $\hat{A}_k$: the quantization-cell intersection logic in Equation~\ref{eqn:abs-reach-set} makes $\hat r_k$ and $\hat c_k$ non-differentiable in $\Theta$. Hence, we replace the abstract reachable sets with smooth approximations:
\begin{equation}
\label{eqn:conc-reach-set}
    A_0 = \Psi_\theta(\hat{x}_0), \; A_{k+1} = \mathrm{Reach}(A_k) \oplus [\shortneg \tilde{e}, \tilde{e}], \; k \ge 0,
\end{equation}
where $\tilde{e}$ is a fixed inflation term intended to cover the volume between $\mathrm{Reach}(A_k)$ and $\Psi_\theta(\hat{A}_{k+1})$. A suitable choice of $\tilde{e}$ yields an over-approximation of $\Psi_\theta(\hat{A}_{k+1})$ by $A_{k+1}$:

\begin{proposition} [Containment of abstract reachable sets]
\label{prop:etilde}
Let $\hat{A}_k$ and $A_k$ denote the abstract and concrete reachable sets from $\hat{x}_0$ and $\Psi_{\theta}(\hat{x}_0)$, respectively, that evolve according to Equations~\ref{eqn:abs-reach-set} and~\ref{eqn:conc-reach-set}. Let $\tilde{e}=[e_1, \dots, e_n]$, where $\tilde{e}_i \ge \max_{j} \eta_{i,j}$. Then, for any $k \ge 0$, it holds that $\Psi_\theta(\hat{A}_k) \subseteq A_k$.
\end{proposition}

\noindent
See the appendix for the proof. This choice of $\tilde{e}$ is sufficient to certify coverage of every concretized abstract reachable set, but may be unnecessarily conservative in practice. Therefore, we treat $\tilde{e}$ as a hyperparameter of the objective.

The second source of non-differentiability is the pair of suprema over $k\in[0,H]$ and $\hat{x}\in\hat{X}$. We replace them with nested, temperature-scaled log-sum-exp (LSE) aggregations. Let $\tau_1$ and $\tau_2$ denote the temperatures of the horizon and abstract-state LSEs, respectively. Let $r_k$ and $c_k$ be the radius and centroid of $\mathrm{Reach}(A_k)$, respectively. The \emph{smooth simulation surrogate} is
\begin{align}
    \tilde{\sigma}^\leftarrow_H(s, \theta) &= \tau_2 \log \sum_{\hat{x} \in \hat{X}} \exp \left( \frac{W_H(\hat{x})}{\tau_2}\right),\\
    \text{where } W_H(\hat{x}) &= \tau_1 \log \sum_{0 \le k \le H} \exp \left( \frac{r_k + \|f^k(c_0) - c_k\|}{\tau_1} \right). \notag
\end{align}
$S^3$ therefore depends on four hyperparameters: (1) the chosen horizon $H$, (2) the inner temperature $\tau_1$ controlling the aggregation of $\tilde{s}_k$ over $k\in[0,H]$; (3) the outer temperature $\tau_2$ controlling the aggregation of $W_H$ over $\hat{x}\in\hat{X}$; and (4) the inflation vector $\tilde{e}$ determining the compounding growth of the reachable sets over $k\in[0,H]$. The experimental section provides guidance on selecting these values.

\subsection{Optimization pipeline}
\label{sec:pipeline}

Recall the optimization problem posed in Subproblem 3:
\begin{equation*}
    \min_{\theta \in \Theta} \tilde\sigma_H^{\leftarrow}(s, \theta) \quad \mathrm{s.t.} \quad \sigma^{\rightarrow}(s, \hat{s}_\theta) = 0.
\end{equation*}
Since the cell boundaries, reachable-set propagation, and log-sum-exp aggregations are differentiable in $\theta$, the gradient $\nabla_\theta \tilde{\sigma}^{\leftarrow}_{H}(s, \theta)$ exists $\forall \theta \in \Theta$. In practice, we compute this via automatic differentiation in Python. We do not make any assumptions regarding the convexity of the objective function $S^3$, so we leverage stochastic gradient descent to tune $\theta$.

To handle the constraint $\sigma^{\rightarrow}(s,\hat{s}_{\theta}) = 0$, we leverage the fact that any $\theta \in \Theta$ constitutes a valid quantizer $\psi_\theta$, and that under this quantization scheme coupled with Taylor model-based transition building, Proposition~\ref{prop:exact-simulation} guarantees $s \lesssim \hat{s}_\theta$ (i.e., $\sigma^{\rightarrow}(s,\hat{s}_{\theta}) = 0$) by construction for any $\theta \in \Theta$.

\section{Experimental Evaluation}
\label{sec:exp}

Our experimental evaluation addresses three questions. First, how accurately does $S^3$ predict the finite-horizon simulation metric? Second, does minimizing $S^3$ improve the resulting abstractions? Third, how sensitive is $S^3$ to its hyperparameters? We address these questions first by testing the validity of $S^3$ as a proxy of $\sigma^{\leftarrow}_H$ with various hyperparameters. Then, we deploy $S^3$ to a stochastic gradient descent pipeline to learn $\theta$, and report the final performance in terms of the ground-truth $\sigma^{\leftarrow}_H$ and \emph{verification recall}. Experiments were conducted on three machines equipped with 12th--14th Gen Intel Core i7/i9 processors (12--24 cores), 32--64 GB of RAM.

\begin{table*}[t]
    \centering
    \caption{Correlation coefficients between $S^3$ and simulation metrics, and their evaluation time on randomly sampled parameters. }
    \label{tab:correlation1}

    \setlength{\tabcolsep}{5pt}
    \renewcommand{\arraystretch}{1.15}
    \small

    \begin{tabular}{
        c
        c
        cc
        cc
        cc
        cc
    }
        \toprule
        & &
        \multicolumn{2}{c}{(Max) $\sigma^\leftarrow_H$} &
        \multicolumn{2}{c}{(Mean) $\overline{\delta}_H$} &
        \multicolumn{2}{c}{(Median) $\delta^{0.50}_H$} &
        \multicolumn{2}{c}{Time per sample (s)} \\

        \cmidrule(lr){3-4}
        \cmidrule(lr){5-6}
        \cmidrule(lr){7-8}
        \cmidrule(lr){9-10}

        & $H$
        & $r_p \uparrow$
        & $r_s \uparrow$
        & $r_p \uparrow$
        & $r_s \uparrow$
        & $r_p \uparrow$
        & $r_s \uparrow$
        & (Real) $T \downarrow$
        & (Proxy) $\hat{T} \downarrow$ \\
        \midrule

        \multirow{5}{*}{\rotatebox[origin=c]{90}{Spiral}}

        & 1
        & $0.88^{+0.03}_{-0.06}$ & $0.85^{+0.06}_{-0.08}$
        & $0.45^{+0.12}_{-0.15}$ & $0.42^{+0.14}_{-0.16}$
        & --- & ---
        & $2.95^{+0.05}_{-0.04}$ & $0.01^{+0.03}_{-0.01}$ \\
        
        & 2
        & $0.78^{+0.10}_{-0.10}$ & $0.70^{+0.10}_{-0.14}$
        & $0.69^{+0.10}_{-0.14}$ & $0.67^{+0.10}_{-0.14}$
        & $0.47^{+0.13}_{-0.19}$ & $0.47^{+0.15}_{-0.19}$
        & $3.08^{+0.04}_{-0.04}$ & $0.01^{+0.03}_{-0.01}$ \\
        
        & 3
        & $0.73^{+0.08}_{-0.11}$ & $0.75^{+0.08}_{-0.11}$
        & $0.66^{+0.09}_{-0.14}$ & $0.66^{+0.11}_{-0.15}$
        & $0.57^{+0.12}_{-0.17}$ & $0.56^{+0.14}_{-0.18}$
        & $3.61^{+0.06}_{-0.05}$ & $0.01^{+0.03}_{-0.01}$ \\
        
        & 4
        & $0.62^{+0.13}_{-0.15}$ & $0.65^{+0.10}_{-0.14}$
        & $0.64^{+0.10}_{-0.15}$ & $0.65^{+0.11}_{-0.15}$
        & $0.56^{+0.12}_{-0.16}$ & $0.59^{+0.12}_{-0.17}$
        & $4.11^{+0.06}_{-0.06}$ & $0.01^{+0.03}_{-0.01}$ \\
        
        & 5
        & $0.52^{+0.17}_{-0.17}$ & $0.55^{+0.13}_{-0.17}$
        & $0.63^{+0.10}_{-0.15}$ & $0.65^{+0.11}_{-0.15}$
        & $0.57^{+0.13}_{-0.19}$ & $0.60^{+0.12}_{-0.16}$
        & $4.30^{+0.06}_{-0.06}$ & $0.01^{+0.03}_{-0.01}$ \\

        \midrule

        \multirow{5}{*}{\rotatebox[origin=c]{90}{Unicycle}}
        & 1
        & $0.84^{+0.05}_{-0.07}$ & $0.85^{+0.05}_{-0.08}$
        & $0.01^{+0.18}_{-0.17}$ & $0.04^{+0.19}_{-0.19}$
        & $\shortneg 0.15^{+0.17}_{-0.16}$ & $\shortneg 0.13^{+0.20}_{-0.18}$
        & $18.37^{+1.28}_{-1.29}$ & $0.01^{+0.05}_{-0.01}$ \\
        
        & 2
        & $0.64^{+0.10}_{-0.12}$ & $0.62^{+0.11}_{-0.15}$
        & $0.58^{+0.12}_{-0.15}$ & $0.54^{+0.14}_{-0.17}$
        & $0.22^{+0.16}_{-0.19}$ & $0.20^{+0.18}_{-0.20}$
        & $21.51^{+1.50}_{-1.52}$ & $0.01^{+0.05}_{-0.01}$ \\
        
        & 3
        & $0.58^{+0.11}_{-0.14}$ & $0.58^{+0.12}_{-0.16}$
        & $0.59^{+0.12}_{-0.15}$ & $0.59^{+0.12}_{-0.16}$
        & $0.50^{+0.15}_{-0.17}$ & $0.48^{+0.14}_{-0.18}$
        & $29.20^{+2.04}_{-2.03}$ & $0.01^{+0.06}_{-0.01}$ \\
        
        & 4
        & $0.46^{+0.11}_{-0.14}$ & $0.45^{+0.13}_{-0.16}$
        & $0.55^{+0.10}_{-0.13}$ & $0.57^{+0.12}_{-0.17}$
        & $0.58^{+0.09}_{-0.12}$ & $0.60^{+0.11}_{-0.15}$
        & $46.74^{+3.23}_{-3.24}$ & $0.01^{+0.06}_{-0.01}$ \\
        
        & 5
        & $0.47^{+0.13}_{-0.16}$ & $0.48^{+0.14}_{-0.17}$
        & $0.50^{+0.11}_{-0.14}$ & $0.53^{+0.13}_{-0.17}$
        & $0.58^{+0.10}_{-0.13}$ & $0.58^{+0.12}_{-0.16}$
        & $84.38^{+5.83}_{-5.56}$ & $0.01^{+0.05}_{-0.01}$ \\

        \midrule

        \multirow{5}{*}{\rotatebox[origin=c]{90}{Mountain Car}}
        & 1
        & $0.68^{+0.11}_{-0.12}$
        & $0.70^{+0.10}_{-0.13}$
        & $\shortneg 0.04^{+0.19}_{-0.20}$
        & $\shortneg 0.01^{+0.19}_{-0.20}$
        & $\shortneg0.33^{+0.19}_{-0.17}$
        & $\shortneg 0.33^{+0.19}_{-0.16}$
        & $2.44^{+0.01}_{-0.01}$
        & $0.02^{+0.07}_{-0.02}$ \\

        & 2
        & $0.70^{+0.10}_{-0.12}$
        & $0.69^{+0.11}_{-0.15}$
        & $0.81^{+0.06}_{-0.08}$
        & $0.81^{+0.06}_{-0.10}$
        & $\shortneg0.21^{+0.19}_{-0.21}$
        & $0.24^{+0.19}_{-0.21}$
        & $3.75^{+0.01}_{-0.01}$
        & $0.01^{+0.07}_{-0.01}$ \\

        & 3
        & $0.69^{+0.08}_{-0.11}$
        & $0.69^{+0.10}_{-0.13}$
        & $0.86^{+0.05}_{-0.08}$
        & $0.86^{+0.05}_{-0.07}$
        & $0.56^{+0.12}_{-0.17}$
        & $0.55^{+0.13}_{-0.18}$
        & $4.86^{+0.03}_{-0.02}$
        & $0.01^{+0.06}_{-0.01}$ \\

        & 4
        & $0.71^{+0.08}_{-0.10}$
        & $0.72^{+0.09}_{-0.12}$
        & $0.84^{+0.05}_{-0.07}$
        & $0.83^{+0.06}_{-0.08}$
        & $0.65^{+0.10}_{-0.12}$
        & $0.64^{+0.11}_{-0.14}$
        & $6.11^{+0.03}_{-0.02}$
        & $0.01^{+0.07}_{-0.01}$ \\

        & 5
        & $0.74^{+0.08}_{-0.12}$
        & $0.72^{+0.10}_{-0.13}$
        & $0.81^{+0.05}_{-0.08}$
        & $0.83^{+0.06}_{-0.09}$
        & $0.72^{+0.08}_{-0.11}$
        & $0.73^{+0.09}_{-0.13}$
        & $7.38^{+0.03}_{-0.03}$
        & $0.01^{+0.07}_{-0.01}$ \\

        \bottomrule
    \end{tabular}
    \vspace{-1mm}
\end{table*}

\subsection{Case studies and baselines}

\noindent \textbf{Case studies.} We evaluate our method on three case studies: (1) a two-dimensional ``spiral'' system; (2) a three-dimensional ``unicycle'' (Dubins') system~\cite{dubins_curves_1957} controlled by a deterministic obstacle-avoidance policy; and (3) the Gymnasium MountainCar-continuous-v0 benchmark~\cite{towers_gymnasium_2025}, controlled by a pretrained deep deterministic policy gradient model~\cite{sb3_ddpg_mountaincarcontinuous_v0}. More details for these systems --- including their dynamical equations and visualizations --- are in the appendix.

\noindent \textbf{Counterexample-guided abstraction refinement (CEGAR).} The first baseline is CEGAR~\cite{clarke_counterexample-guided_2000} --- a well-known abstraction refinement geared for verification tasks. On any model, CEGAR finds abstract trajectories that violate a verification property, and then check if they are spurious (the concrete model satisfies the property). If spurious, CEGAR iteratively refines quantization cells along the trajectory until either the counterexample is eliminated, or the permitted number of refinements is exhausted. For comparability to our models of grid resolution $m_i$, we build a slightly smaller, uniform-grid model with resolution $m_i - \kappa_i$ so that the allocated refinement budget is $\prod_{i=1}^nm_i - \prod_{i=1}^nm_i\kappa_i$.

\noindent \textbf{Grid resolution optimization baseline.} The second baseline is the grid-resolution optimization method~\cite{weber_optimized_2017}. Here, a convex optimization problem is solved to minimize the \emph{expected number of successors} in $\hat{f}$ by tuning the common cell aspect ratio, subject to a fixed volume constraint. For comparability to our models of resolution $m_i$, we set this fixed volume to $\prod_{i=1}^n(\overline{x}_i - \underline{x}_i)/m_i$ so that the model will possess roughly the same number of states.

\subsection{Proxy validation}

\noindent \textbf{Evaluating the real finite-horizon simulation metric.}
For a given abstract model $\hat{s}_\theta$ and dynamical system $s$, we compute the finite-horizon reverse simulation metric $\sigma_H^\leftarrow$ by directly evaluating Equation~\ref{eqn:simulation-from-bellman}. Specifically, for each $\hat{x} \in \hat{X}$, we solve $\delta_H(\hat{x}) = \inf_{x \in X} V_H(x,\hat{x})$ (via Powell’s method~\shortcite{powell1964efficient} in SciPy's optimization library) and then take $\sigma_H^\leftarrow = \max_{\hat{x} \in \hat{X}} \delta_H(\hat{x})$. This maximum is the smallest $\delta \ge 0$ satisfying the conditions of finite-horizon $\delta$-approximate reverse simulation, per Definition~\ref{def:reverse-sim} and Equation~\ref{eqn:reverse-sim-metric}. We also consider additional statistics of state-wise values $\delta_H(\hat{x})$, including the mean $\overline{\delta}_H$ over $\hat{X}$ and median $\delta_H^{0.50}$. Unlike the worst-case metric $\sigma_H^\leftarrow$, these statistics summarize abstraction-wide discrepancies and thus provide a broader characterization of abstraction quality.

\noindent \textbf{Experimental setup.} To validate the quality of our proxy, we robustly evaluate the correlation between $S^3$ and the finite-horizon reverse simulation metric, together with its aggregate statistics. For each $H \in \{1,\dots,5\}$, we draw $100$ parameter vectors $\theta$ from a standard normal distribution using JAX NumPy with random seed $0$. For each parameter vector, we construct an abstract model and evaluate the reverse simulation metric of $\hat{s}_\theta$ by $s$. We record the runtime $T$ of this evaluation. We evaluate $S^3$ on these parameters and record both the proxy value and its evaluation time $\hat{T}$. We report $95\%$ bootstrap confidence intervals for the correlation and runtime statistics, computed with SciPy using the robust bias-corrected and accelerated method with $20{,}000$ resamples.

Table~\ref{tab:correlation1} reports the Pearson ($r_p$) and Spearman ($r_s$) correlation between $S^3$ and each reverse simulation statistic.  Here, the hyperparameters are $\tau_1=\tau_2=0.1$ and $\tilde{e}$ set to half of the average side length of the cells (e.g., for the spiral with a domain side length of $10.0$ along the first dimension, we set $\tilde{e}_1 = 10.0/(2 m_1)$, where $m_1$ is the grid resolution).

\noindent \textbf{$S^3$ exhibits high empirical correlation with the finite-horizon simulation metric.} A strong positive correlation suggests that minimizing $S^3$ during optimization will also reduce the underlying simulation metric. As shown in Table~\ref{tab:correlation1}, $S^3$ is strongly correlated with $\sigma_H^\leftarrow$, particularly at shorter horizons. This correlation weakens as $H$ increases, likely due to compounding over- and under-approximation in the smooth reachable-set recursion of Equation~\ref{eqn:conc-reach-set}.


\looseness=-1
\noindent \textbf{Temperature trades off max vs mean correlation.} Per Table~\ref{tab:correlation1}, $S^3$ exhibits moderate positive correlation with the mean and median $\delta$ values, particularly at intermediate and longer horizons. The exception is $H=1$, where both statistics have near-zero variance across the sampled parameter vectors, leaving insufficient dispersion to support a meaningful correlation estimate. For Spiral, the median is effectively constant, which also prevents construction of our confidence interval. This behavior reflects the choice $\tau_1=\tau_2=0.1$, which makes the LSE closely approximate the suprema in Equation~\ref{eqn:abs-value-fcn} and Proposition~\ref{prop:upper-bound}. In the appendix, we repeat the analysis with larger outer temperatures and find that correlation with the mean and median $\delta$ values improves at the expense of correlation with the worst-case metric $\sigma_H^\leftarrow$.

\noindent \textbf{$S^3$ is much faster to evaluate than the finite-horizon simulation metric.} Evaluating $S^3$ requires only $0.01$--$0.02$ seconds per parameter vector across all case studies and horizons, making it practical for gradient-based optimization. By comparison, evaluating the $\sigma_H^\leftarrow$ requires $2.95$--$4.30$ seconds for Spiral, $18.4$--$84.4$ seconds for Unicycle, and $2.44$--$7.38$ seconds for Mountain Car. This cost may be prohibitive for sample-intensive, gradient-free optimization methods.

\subsection{Optimization performance}

\noindent \textbf{Model checking task.} To examine verification utility, we perform model checking using PyModelChecking~\cite{casagrande_albertocasagrandepymodelchecking_2025} on the optimized models. The checked properties are specific to the goal of each environment (see appendix for LTL formulae). For spiral, we check whether it reaches the goal set at its equilibrium and remains in bounds. For unicycle, we check whether it reaches a goal set while avoiding an obstacle and staying within bounds. For mountain car, we check whether it reaches the objective under the pretrained DDPG. We report model checking performance through \emph{recall}: the verified volume of the state space divided by the total volume of ground-truth satisfying state space (the region where the concrete system satisfies the property).

\noindent \textbf{Experimental setup.} To evaluate $S^3$ as an objective for optimizing the quantization parameters, we consider two abstraction sizes for each case study. For the two-dimensional Spiral and Mountain Car systems, we use $m_i\in \{70,100\}$, yielding $4{,}900$ and $10{,}000$ abstract states, respectively. For the three-dimensional Unicycle system, we use $m_i\in\{50,100\}$, yielding $125{,}000$ and $1{,}000{,}000$ states. The initial parameters are drawn from a zero-mean Gaussian distribution of variance $0.1$. We use the hyperparameters from Table~\ref{tab:correlation1} and horizon $H=3$, and $\kappa_i=10$. Model building time is denoted $T_b$ (abstraction and optimization or refinement), verification (model checking) time $T_v$, and their sum is denoted $T_t$. 

\begin{table}

    \centering
    \caption{Model performance across three methods: (a) CEGAR, (b) Weber~\shortcite{weber_optimized_2017}, (c) $S^3$ + SGD (ours).}
    \label{tab:optimization-results-1}
\vspace{-2mm}
    \setlength{\tabcolsep}{2.45pt}
    \renewcommand{\arraystretch}{1.0}
    \footnotesize

    \begin{tabular}{
        c
        c
        c
        ccc
        c
        ccc
    }
        \toprule
        & & &
        \multicolumn{3}{c}{Sim. Metrics} &
        \multicolumn{1}{c}{Task} &
        \multicolumn{3}{c}{Time (s)} \\

        \cmidrule(lr){4-6}
        \cmidrule(lr){8-10}

        & $m_i$
        &
        & $\sigma^\leftarrow_H \downarrow$
        & $\overline{\delta}_H \downarrow$
        & $\delta_H^{0.50} \downarrow$
        & Recall $\uparrow$
        & $T_b \downarrow$
        & $T_v \downarrow$
        & $T_t \downarrow$ \\
        \midrule

        \multirow{6}{*}{\rotatebox[origin=c]{90}{Spiral}}
        & \multirow{3}{*}{70}
        & (a)
        & \textbf{2.020} & \textbf{0.373} & \textbf{0.388}
        & \textbf{0.916}
        & 8.520 & 0.032 & 8.552 \\

        & & (b)
        & 2.064 & 0.372 & 0.429
        & 0.872
        & 0.140 & 0.030 & 0.170 \\

        & & (c)
        & 2.064 & 0.372 & 0.429
        & 0.872
        & \textbf{0.130} & \textbf{0.020} & \textbf{0.150} \\

        \cmidrule(lr){2-10}

        & \multirow{3}{*}{100}
        & (a)
        & \textbf{2.120} & \textbf{0.264} & \textbf{0.281}
        & \textbf{0.939}
        & 24.14 & 0.045 & 24.19 \\

        & & (b)
        & 2.130 & 0.268 & 0.300
        & 0.928
        & \textbf{0.250} & 0.090 & 0.340 \\

        & & (c)
        & 2.130 & 0.268 & 0.300
        & 0.932
        & 0.270 & \textbf{0.030} & \textbf{0.300} \\

        \midrule

        \multirow{6}{*}{\rotatebox[origin=c]{90}{Unicycle}}
        & \multirow{3}{*}{50}
        & (a)
        & $10.19$& $2.253$& $2.162$& 0.234
        & 1921 & 3.213 & 1924 \\

        & & (b)
        & \textbf{9.410} & 2.282 & 2.340
        & 0.220
        & 140.1 & 3.030 & 143.1 \\

        & & (c)
        & 9.970 & \textbf{1.930} & \textbf{1.990}
        & \textbf{0.249}
        & \textbf{138.3} & \textbf{2.220} & \textbf{140.2} \\

        \cmidrule(lr){2-10}

        & \multirow{3}{*}{100}
        & (a)
        & $9.350$& $1.060$& $1.110$& 0.452
        & 4246 & \textbf{16.45} & 4262 \\

        & & (b)
        & \textbf{9.104} & 1.101 & 1.160
        & \textbf{0.806}
        & 1105 & 23.30 & 1128 \\

        & & (c)
        & 10.15 & \textbf{0.933} & \textbf{0.977}
        & 0.786
        & \textbf{1079} & 19.20 & \textbf{1098} \\

        \midrule

        \multirow{6}{*}{\rotatebox[origin=c]{90}{Mountain Car}}
        & \multirow{3}{*}{70}
        & (a)
        & 0.079 & 0.040 & 0.035
        & 0.130
        & 96.91 & 0.152 & 97.05 \\

        & & (b)
        & \textbf{0.070} & 0.035 & 0.028
        & 0.172
        & 13.82 & \textbf{0.030} & 13.83 \\

        & & (c)
        & 0.074 & \textbf{0.032} & \textbf{0.027}
        & \textbf{0.234}
        & \textbf{12.04} & 0.050 & \textbf{12.05} \\

        \cmidrule(lr){2-10}

        & \multirow{3}{*}{100}
        & (a)
        & \textbf{0.067} & 0.026 & 0.022
        & \textbf{0.436}
        & 123.3 & 0.178 & 123.5 \\

        & & (b)
        & 0.068 & \textbf{0.024} & \textbf{0.021}
        & 0.276
        & \textbf{21.10} & \textbf{0.060} & \textbf{21.16} \\

        & & (c)
        & 0.069 & \textbf{0.024} & 0.022
        & 0.300
        & 22.80 & \textbf{0.060} & 22.86 \\

        \bottomrule
    \end{tabular}
\end{table}

\noindent \textbf{$S^3$ scales better than CEGAR.} Across all case studies, $S^3$ with SGD requires far less total computation than CEGAR. For Spiral, total runtime is $0.15$--$0.3$ seconds for $S^3$ versus $8.55$--$24.2$ seconds for CEGAR; for Mountain Car, $12.1$--$22.9$s versus $97.1$--$123.5$s. The largest gap occurs on unicycle, where $S^3$ requires $140$ and $1098$ seconds, compared with $1924$ and $4262$ seconds for CEGAR. This because we optimize a differentiable surrogate, rather than repeatedly rebuilding, checking, and refining the abstraction like CEGAR.

\noindent \textbf{The methods offer complementary parameterizations.} Weber optimizes dimension-wise resolution through a common aspect ratio, CEGAR locally splits selected cells, and $S^3$ optimizes individual gap widths at fixed resolution. Despite these differences, $S^3$ remains competitive throughout. At $m_i=100$ on Spiral, CEGAR slightly improves $\sigma_H^\leftarrow$ and recall ($2.12$, $0.939$) over $S^3$ ($2.13$, $0.932$). On Mountain Car with $m_i=70$, Weber attains the lowest worst-case metric ($0.070$ versus $0.074$), while $S^3$ achieves lower mean and median discrepancies ($0.032$, $0.027$) and higher recall ($0.234$). Therefore, these methods could be composed: Weber can select the dimension-wise resolution, $S^3$ can optimize the gaps, and CEGAR can apply final task-specific refinement.

\looseness=-1
\noindent \textbf{$S^3$ improves abstraction-wide quality.} On Unicycle, $S^3$ improves the mean and median discrepancies even when the worst-case metric is slightly larger than Weber's. At $m_i=50$, it reduces the mean from $2.28$ to $1.93$ and the median from $2.34$ to $1.99$, while increasing recall from $0.220$ to $0.249$. At $m_i=100$, it similarly reduces the mean from $1.10$ to $0.933$ and the median from $1.16$ to $0.977$. This agrees with the proxy-validation results. 


\noindent \textbf{Future work.} Several extensions could improve the tightness and flexibility of $S^3$. First, jointly optimizing the $\arg\inf_{x \in X}$ in Equation~\ref{eqn:simulation-from-bellman} with the quantization parameters rather than arbitrarily selecting the centroid $c_0$ (see the proof of Proposition~\ref{prop:upper-bound}) can tighten the upper-bounding value function (Equation~\ref{eqn:abs-value-fcn}). Second, more expressive quantization schemes like neural parameterizations could represent partitions beyond variable-width grids. Extending the framework to probabilistic transition systems is also an open direction.

\bibliography{ref.bib}


\newpage

\section{Proofs of Propositions}

Throughout these proofs, we denote the concretization of a set of abstract states $\hat{A} \subseteq \hat{X}$ as the union of the concretized abstract states in the set: $\Psi_\theta(\hat{A}) = \bigcup_{\hat{x} \in \hat{A}}\Psi_\theta(\hat{x})$. 

\setcounter{proposition}{0}
\setcounter{figure}{0}

\begin{proposition}[Exact simulation]
    Given a dynamical system $s = (X, X_0, f)$, a parameter vector $\theta \in \Theta$, and a rectilinear state quantization function $\psi_\theta(x)$. Let the abstract state set be $\hat{X} = \psi_\theta(X)$, the initial abstract state set be $\hat{X}_0 = \psi_\theta(X_0)$, and let $\hat{f}(\hat{x})$ satisfy Equation~\ref{eqn:post-set} for all $\hat{x} \in \hat{X}$, constituting $\hat{s} = (\hat{X}, \hat{X}_0, \hat{f})$. Then, it follows that $\hat{s}_\theta$ simulates $s$ exactly, that is, $s \lesssim \hat{s}$.
\end{proposition}

\begin{proof} [Proof of Proposition~\ref{prop:exact-simulation}]
     Recall the distance metric 
    \begin{equation*}
        d(x, \hat{x})=\inf_{z \in \Psi_\theta(\hat{x})} \|x-z\|_p, \quad p \ge 0.
    \end{equation*}
    For every $(x, \hat{x}) \in R_\varepsilon$, two conditions must be met for exact ($\varepsilon=0$) forward simulation to hold (see Definition~\ref{def:forward-sim}). First, the distance from $x$ to $\hat{x}$ must be $0$, that is, $d(x, \hat{x})=0$. Second, there must exist some $\hat{x}' \in \hat{f}(\hat{x})$ such that $(f(x), \hat{x}')\in R_0$ (recursively satisfying $d(f(x), \hat{x}')=0$ and $\exists \hat{x}'' \in \hat{f}(\hat{x}')$ where $(f(f(x)), \hat{x}'')\in R_0$, and so on). Since $d(x, \hat{x})$ is the smallest distance from $x$ to any point in $\Psi_\theta(\hat{x})$, the set of $x$ where $d(x, \hat{x}) = 0$ is exactly the quantization cell $\Psi_\theta(\hat{x})$. 
    
    Furthermore, we know that any concrete point within $\Psi_\theta(\hat{x})$ maps to $\mathrm{Reach}(\Psi_\theta(\hat{x}))$, and that $\mathrm{Reach}(\Psi_\theta(\hat{x}))$ is inscribed by the set of (concretized) successors $\Psi_\theta(\hat{f}(\hat{x}))$ by Equation~\ref{eqn:post-set}, such that $f(\Psi_\theta(\hat{x})) \subseteq \Psi_\theta(\hat{f}(\hat{x}))$. Hence, we have shown that there must exist some $\hat{x}' \in \hat{f}(\hat{x})$ such that $d(f(x), \hat{x}')=0$, satisfying the second condition. 
    
    Since we have proven that these conditions hold on the arbitrary initial pair $(x, \hat{x}) \in R_0$, by induction, we have proven that it holds on every pair thereafter. Therefore, we have proven that $\hat{s}_\theta$ simulates $s$ exactly, that is, $s \lesssim \hat{s}$.
\end{proof}

\begin{proposition}[Upper bound on $\sigma_H^\leftarrow$]
Given $\hat{W}_H(\hat{x})$ per Equation~\ref{eqn:abs-value-fcn}, it follows that $\sigma_H^\leftarrow \le \sup_{\hat{x} \in \hat{X}} \hat{W}_H(\hat{x})$.
\end{proposition}

\begin{proof} [Proof of Proposition~\ref{prop:upper-bound}]
    Recall the distance metric 
    \begin{equation*}
        d(x, \hat{x})=\inf_{z \in \Psi_\theta(\hat{x})} \|x-z\|_p, \quad p \ge 0.
    \end{equation*}
    Now, choose any $\hat{x}_0 \in \hat{X}$. Let $c_0$ be the centroid of $\Psi_\theta(\hat{x})$. By soundness through the construction of $\hat{f}$ according to Equation~\ref{eqn:post-set}. It follows that $\hat{x}_{k+1} \in \hat{f}(\hat{x}_k)$. Then, for every $k$, $\hat{x}_k \in \hat{A}_k$ (Equation~\ref{eqn:abs-reach-set}) and $\Psi_\theta(\hat{x}_k) \subseteq \Psi_\theta(\hat{A}_k)$. Let $\hat{c}_k$ and $\hat{r}_k$ be the centroid and radius of $\Psi_\theta(\hat{A}_k)$, respectively. Then, any $y_k$ in $\Psi_\theta(\hat{x}_k)$ must lie within $\hat{r}_k$ of $\hat{c}_k$, i.e., $\|y_k - \hat{c}_k\| \le \hat{r}_k$. Therefore, it follows that
    \begin{align*}
        d(f^k(c_0), \hat{x}_k) &\le \inf_{y \in \Psi_\theta(\hat{x}_k)} \|f^k(c_0) - y \| \\
        &\le \|f^k(c_0) - y_k\| \\
        &\le \|f^k(c_0) - \hat{c}_k\| + \|\hat{c}_k - y_k\| \\
        &\le \|f^k(c_0) - \hat{c}_k\| + \hat{r}_k.
    \end{align*}
    Taking the supremum over the finite horizon $0 \le k \le H$ gives:
    \begin{equation*}
        V_H(c_0, \hat{x}_o) \le \sup_{0 \le k \le H} \{ \|f^k(c_0) - \hat{c}_k\| + \hat{r}_k \} = \hat{W_H}(\hat{x}).
    \end{equation*}
    Since $c_0 \in \Psi_\theta(\hat{x}_0) \subseteq X$, it follows that:
    \begin{equation*}
        \inf_{x \in X} V_H(x, \hat{x}_0) \le V_H(c_0, \hat{x}_o) \le \hat{W}_H(\hat{x}).
    \end{equation*}
    And finally, by taking the supremum over $\hat{x}_0 \in \hat{X}$, we see:
    \begin{equation*}
        \sigma_H^\leftarrow(s, \hat{s}_\theta) \le \sup_{\hat{x}_0 \in \hat{X} }\inf_{x \in X} V_H(x, \hat{x}_0) \le \sup_{\hat{x}_0 \in \hat{X} } \hat{W}_H(\hat{x}).
    \end{equation*}
    Hence, we have proven the claim directly.
\end{proof}

\begin{proposition} [Containment of abstract reachable sets]
Let $\hat{A}_k$ and $A_k$ denote the abstract and concrete reachable sets from $\hat{x}_0$ and $\Psi_{\theta}(\hat{x}_0)$, respectively, that evolve according to Equations~\ref{eqn:abs-reach-set} and~\ref{eqn:conc-reach-set}. Let $\tilde{e}=[e_1, \dots, e_n]$, where $\tilde{e}_i \ge \max_{j} \eta_{i,j}$. Then, for any $k \ge 0$, it holds that $\Psi_\theta(\hat{A}_k) \subseteq A_k$.
\end{proposition}

\begin{proof} [Proof of Proposition~\ref{prop:etilde}]
    For any $\hat{x} \in \hat{X}$, define $B \subseteq X$ such $B \cap \Psi_\theta(\hat{x}) \ne \emptyset$. Pick any $x \in \Psi_\theta(\hat{x})$ and $y \in B \cap \Psi_\theta(\hat{x})$. Then, for each dimension $i \in \{1, \dots, n\}$, the absolute distance between $x_i$ and $y_i$ is bounded by that cell length of $\hat{x}$, which is bounded by the maximum cell length, i.e. $|x_i - y_i| \le \eta_{i,\hat{x}_i} \le \max_{j} \eta_{i,j}$. Hence, for $\tilde{e}=[e_1, \dots, e_n]$, where $\tilde{e}_i \ge \max_{j} \eta_{i,j}$, it follows:
    \begin{equation}
    \label{eqn:inc-under-inflation}
        \bigcup_{\substack{\hat{x} \in \hat{X} \\ \Psi_\Theta(\hat{x}) \cap B}} \Psi_\theta(\hat{x}) \subseteq B \oplus [-\tilde{e}, \tilde{e}]
    \end{equation}
    We now prove the claim that $\Psi_\theta(\hat{A}_k) \subseteq A_k$ by induction on $k \ge 0$. First, recall that $\Psi_\theta(\hat{A}_0)=\Psi_\theta(\hat{x}_0) = A_0$. Suppose that $\Psi_\theta(\hat{A}_k) \subseteq A_k$. Then, by Equation~\ref{eqn:abs-reach-set} and Equation~\ref{eqn:inc-under-inflation}:
    \begin{align*}
        \Psi_{\theta}(\hat{A}_{k+1}) &= \Psi_\theta(\{\hat{x} \in \hat{X} \mid \Psi_\theta(\hat{x}) \cap\mathrm{Reach}(\Psi_\theta(\hat A_k)) \ne \emptyset \}) \\
        &\subseteq \mathrm{Reach}(\Psi_\theta(\hat{A}_k)) \oplus [-\tilde{e}, \tilde{e}] \\
        &\subseteq \mathrm{Reach}(A_k) \oplus [-\tilde{e}, \tilde{e}] \\
        &= A_{k+1}.
    \end{align*}
    Hence, we have proven that $\Psi_\theta(\hat{A}_k) \subseteq A_k$ for $k \ge 0$.
\end{proof}



\section{Additional System Descriptions}

\subsubsection{Spiral.} In our experiments, we utilize a ``spiral'' case study that models a linear, time-invariant, discrete-time, and globally exponentially stable (GES) system. This two-dimensional system lives in the domain $X= [-10, 10] \times [-10, 10]$, and evolves according to
\begin{equation*}
    x[k+1] = x^{*} + 
    \begin{bmatrix}
        0.8 & -0.3 \\
        0.3 & 0.8
    \end{bmatrix}
    (x[k] - x^{*}), \,\,
    \,\, x^{*} = 
    \begin{bmatrix}
        5.0 \\
        5.0
    \end{bmatrix}.
\end{equation*}
If during simulation the system's state $x$ hits the boundary of $X$, then the system has crashed, and the process is terminated. In verification, we check whether the system reaches within $2.0$ units of the equilibrium $x^*$ and avoids crashing. In linear temporal logic, this is written:
\begin{equation*}
    \varphi := (x\in X) ~ \mathcal{U}\left[\|x - x^*\| \le 2.0\right]
\end{equation*}
Furthermore, we define $X_0 = X$. 

\begin{figure}[h]
    \label{fig:spiral-disp-field}
    \centering
    \includegraphics[width=0.9\linewidth]{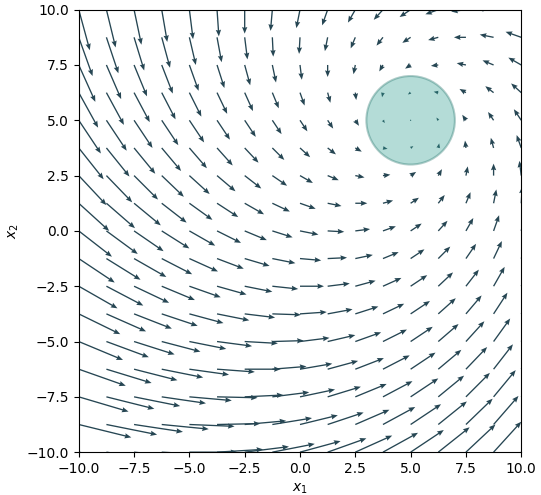}
    \caption{Displacement field of the spiral system. Blue circle is the goal set.}
\end{figure}

\subsubsection{Autonomous Unicycle.} The unicycle dynamical system in our experiments evolves according to Dubins path constraints. This three-dimensional system lives in the domain $X =  [0.0, 50.0] \times [0.0, 50.0] \times [-\pi, \pi]$, and we let $X_0 =  [0.0, 50.0] \times [0.0, 50.0] \times [-\pi/2, \pi/2]$. The system evolves according to:
\begin{equation*}
    f(x)
    =
    \begin{bmatrix}
        x_1[k] \\
        x_2[k] \\
        x_3[k]
    \end{bmatrix}
    +
    \begin{bmatrix}
        v\cos x_3[k] \\
         v\sin x_3[k] \\
         u[k]
    \end{bmatrix}
    \Delta t,
\end{equation*}
where $\Delta t = 0.5$ is a fixed sampling time, $v = 5.0$ is the constant unicycle speed, and $u$ is the control action (heading rate of change) selected from a deterministic state controller designed for goal-attraction and obstacle-repulsion. The objective of this system is to reach a cylindrical goal (centered at $y^{g} = (x^g_1=40.0, x^g_2=20.0)$ with radius $r^{\mathrm{g}} = 8.0$) while avoiding a cylindrical obstacle (centered at $y^{\mathrm{o}} = (x^o_1=25.0, x^o_2=25.0)$ with radius $r^{o} = 5.0$) and the boundary of $X$. In linear temporal logic, this property is written:
\begin{equation*}
   \varphi := [x\in X \wedge d_1(x) > 5] ~\mathcal{U}~(d_2(x) \le 2),
\end{equation*}
where $d_1(x) = \sqrt{(x_1-25)^2+(x_2-25)^2}$ and $d_2(x) = \sqrt{(x_1-40)^2+(x_2-20)^2}$.

The unicycle controller used in our experiments is a smooth, deterministic, state-feedback law that selects the angular velocity $u$ by steering the system toward a desired heading $\alpha_d(x)$ computed from the current state. At a high level, the controller forms a planar guidance vector $v$ by combining an \emph{attractive} vector $v^{\mathrm{att}}$ that points toward the goal with a \emph{repulsive} vector $v^{\mathrm{rep}}$ that pushes the system away from the obstacle. The resulting vector determines the desired heading $\alpha_d$, which is used to determine the control action (heading rate of change) $u$.

Let $y = (x_1,x_2)$ denote the positional component of the state $x$, and let $ d^{\mathrm{obs}}(x) = y - y^{\mathrm{o}}$ be the displacement vector. Its Euclidean norm $\|d^{\mathrm{obs}}(x)\|$ is the distance from the system to the obstacle center. From this, we calculate an exponentially activated clearance $\tilde{c}$ from the absolute clearance $c$ from the system's position and the obstacle boundary:
\begin{equation*}
   \tilde{c}(x) = e^{-\alpha c(x)}, \quad  c(x) = \|d^{\mathrm{obs}}(x)\| - r^{\mathrm{obs}},
\end{equation*}
where $\alpha \in \mathbb{R}_{>0}$ controls the rate of decay. Thus, when the system lies on the obstacle boundary ($c(x)=0$), we have $\tilde{c}(x)=1$, and this quantity decays toward $0$ as the system moves away from the obstacle.

A repulsive vector is computed as
\begin{equation*}
    v^{\mathrm{rep}}(x) =
    \left[\frac{\tilde{c}(x)}{\|d^{\mathrm{obs}}(x)\|^3 + \epsilon}\right] d^{\mathrm{obs}}(x),
\end{equation*}
where $\epsilon \in \mathbb{R}_{>0}$ is a small constant included to avoid singular behavior near the obstacle center. The attractive vector is then computed $v^{\mathrm{att}}(x) = y^{\mathrm{goal}} - y$. Then, $ v^{\mathrm{rep}}$ and $v^{\mathrm{att}}(x)$ are combined via a weighted sum $ v(x) = k^{\mathrm{rep}} v^{\mathrm{rep}}(x) + k^{\mathrm{att}} v^{\mathrm{att}}(x)$, where $k^{\mathrm{rep}}$ and $k^{\mathrm{att}}$ are repulsion and attraction gains, respectively. Intuitively, when the system is far from the obstacle, $v^{\mathrm{rep}}(x)$ is small and the heading is dominated by attraction toward the goal; near the obstacle, the repulsive term grows and steers the system away from collision.

Finally, desired heading is computed by $ \alpha_d(x) = \arctan(v_1(x), v_2(x))$, and the control input is determined by $ u(x) = u_{\max}\tanh(k_{\alpha}(\alpha_d(x)-x_3))$
and the control input is chosen as where $u_{\max}$ sets the allowable turn rate (so that $u \in [-u_{\max},u_{\max}]$) and $k_{\alpha}$ is a tunable gain on the heading error.

For our experiments, we select the following controller parameters. We set the decay factor $\alpha=0.6$ and $\epsilon=10^{-6}$. The controller prioritizes obstacle-avoidance over goal-reaching by letting $k^{\mathrm{att}}=1.0$ and $k^{\mathrm{rep}}=8.0$. Furthermore, the controller gain $k_{\alpha}=2.5$, and $u_{\max}=\pi/4$.

\begin{figure}[h]
    \centering
    \includegraphics[width=0.9\linewidth]{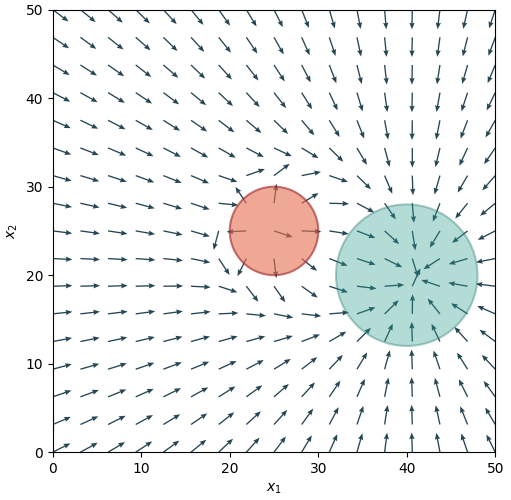}
    \caption{Displacement field of the unicycle system for the $x_3=0$ (heading angle) slice. Red circle is the obstacle, and blue circle is the goal set.}
    \label{fig:placeholder}
\end{figure}

\subsubsection{Mountain Car.} We utilize the Gymnasium MountainCar-continuous-v0 benchmark in our evaluation~\cite{towers_gymnasium_2025}. This environment models a two-dimensional control problem in $X = [-1.2, 0.6] \times [-0.07,0.07]$, where $x_1$ is the horizontal position of the vehicle and $x_2$ is its velocity. The system evolves according to:
\begin{align*}
    x_2[k+1] &= x_2[k] + 0.0015u[k] - 0.0025\cos 3x_1[k] \\
    x_1[k+1] &= x_2[k+1] + x_1[k]
\end{align*}
At each step, the agent selects from actions $u \in [-1, 1]$, corresponding respectively to the application of a negative force (throttle leftward), no force, or a positive force (throttle rightward). The control objective of the Mountain Car is to reach the top of a steep incline ($x_1 \ge 0.45$). In linear temporal logic, this propery is written
\begin{equation*}
    \varphi := \Diamond (x_1 \ge 0.45).
\end{equation*}
For our experiments, we utilize a pretrained deep deterministic policy gradient (DDPG) from Stable-Baselines3~\shortcite{sb3_ddpg_mountaincarcontinuous_v0}.


\begin{figure}[h]
    \centering
    \includegraphics[width=0.9\linewidth]{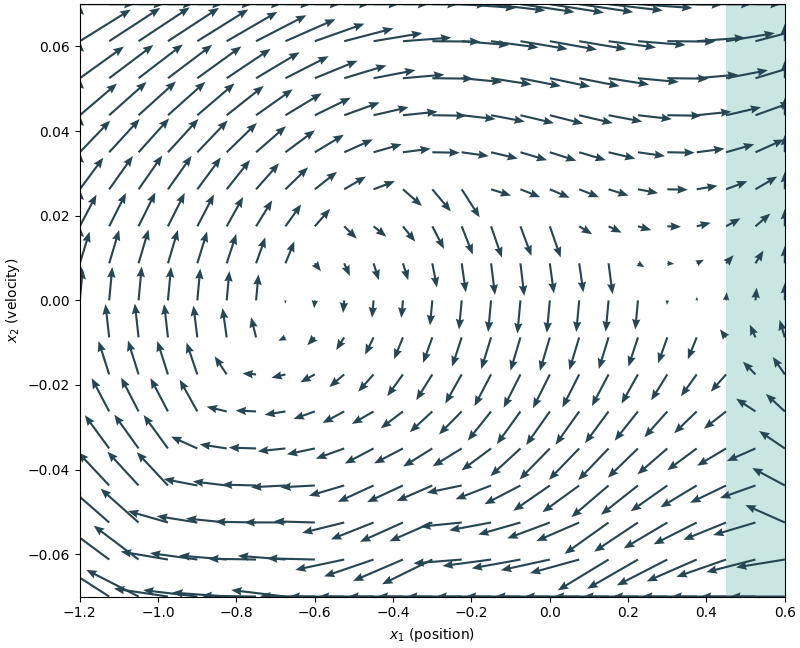}
    \caption{Displacement field of the Mountain Car under the pretrained DDPG. Blue region is the goal set ($x_1 \ge 0.45$).}
    \label{fig:placeholder}
\end{figure}

\section{Algorithms to Build Abstractions}

To construct a finite abstraction $\hat{s}$ of a continuous-state, discrete-time dynamical system $s$ for a given quantization parameter vector $\theta$, we combine rectilinear state quantization with Taylor model-based reachability. Abstract transitions are assigned according to Equation~\ref{eqn:post-set}, yielding an abstraction that satisfies the forward-simulation condition by construction, as established in Proposition~\ref{prop:exact-simulation}.

As a preprocessing step, we derive the Jacobian and Hessian of the closed-loop transition map $f$. For sufficiently simple systems, these derivatives may be obtained manually. In our implementation, we use SymPy to symbolically derive, simplify, and cache the corresponding analytical expressions. Precomputing these derivatives substantially reduces the cost of evaluating them repeatedly across the quantization cells. In our early implementations, computing the derivatives in the loop using automatic differentiation in JAX introduced excessive computational overhead.

Given $\theta$, Algorithm~\ref{alg:getcelledges} maps the unconstrained gap weights to positive cell widths using the normalized Softplus parameterization and computes the corresponding grid boundaries. The Cartesian product of the resulting one-dimensional intervals defines the abstract state set $\hat{X}$. The initial abstract state set $\hat{X}_0$ is then obtained by selecting all quantization cells whose concretizations intersect the initial set $X_0$.

\begin{algorithm}
\caption{\texttt{CellEdges}}
\begin{algorithmic}[1]
\Require Gap weights $\omega$; lower bound $\underline{x}$ and upper bound $\overline{x}$
\State $m \gets \mathrm{len}(\omega)$, $b \gets \langle \underline{x} \rangle$
\For{$j = 2, \dots, m$}
    \State $\eta \gets (\overline{x} - \underline{x})\frac{\mathrm{Softplus}(\omega_{j})}{\sum_{j=1}^{m_i}\mathrm{Softplus}(\omega_{ij})}$
    \State \Call{Append}{$b, b_{j-1}+\eta$}
\EndFor
\State \Return $b$
\end{algorithmic}
\label{alg:getcelledges}
\end{algorithm}

For each abstract state $\hat{x}$, we compute the center and vertices of its associated quantization cell. The vertices are propagated through the first-order Taylor approximation of $f$ about the cell center. Because this approximation is affine, its coordinatewise extrema over the cell are attained at the vertices, and the transformed vertices therefore determine an axis-aligned bounding box for the linearized image. Algorithm~\ref{alg:taylorremainder} computes a componentwise bound on the Taylor remainder using bounds on the Hessian over the cell. To avoid searching for the maximum Hessian spectral norm within each quantization cell, we derive the \emph{interval Hessian} $[H]$. This provides an analytic expression for the component-wise lower and upper bounds of the Hessian evaluated within an AABB. The computed remainder bound is used to inflate the linearized image, producing a conservative over-approximation of the one-step reachable set.

\begin{algorithm}
\caption{\texttt{TaylorRemainder}}
\begin{algorithmic}
    \Require Analytic interval Hessian $[H] = [H^{\min}, H^{\max}]$; AABB $[\ell, u]$; state dimensions $n$
    \Ensure$ e \in \mathbb{R}^n$
    
    \State $h \gets \frac{1}{2}(u-\ell)$
    \State $e \gets 0 \in \mathbb{R}^n$
    
    \For{$i=1,\dots,n$}
        \For{$j=1,\dots,n$}
            \For{$k=1,\dots,n$}
                \State $M_{ijk} \gets \max\left\{
                    |H^{\min}_{ijk}|,
                    |H^{\max}_{ijk}|
                \right\}$
                \State $e_i \gets e_i + \frac{1}{2} M_{ijk} h_j h_k$
            \EndFor
        \EndFor
    \EndFor
    
    \State \Return $e$
\end{algorithmic}
\label{alg:taylorremainder}
\end{algorithm}


Finally, we identify all quantization cells whose concretizations intersect the inflated reachable-set approximation and assign the corresponding abstract states as successors of $\hat{x}$. Algorithm~\ref{alg:buildabstractmodel} applies this procedure to every abstract state and returns the abstract state set $\hat{X}$, initial abstract state set $\hat{X}_0$, and abstract dynamics $\hat{f}$.

\begin{algorithm}[h]
\caption{Construct an abstraction of a dynamical system}
\begin{algorithmic}[1]
\Require Concrete system
$s=(X,X_0,f)$, where
$X=\prod_{i=1}^{n}[\underline{x}_i,\overline{x}_i]$ and
$X_0=\prod_{i=1}^{n}[\underline{y}_i,\overline{y}_i]$;
quantization parameters
$\theta=\{\omega_1,\ldots,\omega_n\}$

\State
\State 
\Comment{Derive Taylor model/remainder components}

\State $J \gets \bigl(z \mapsto\left.\frac{\partial f}{\partial x}\right|_{x=z}\bigr)$, $[H] \gets \bigl(z \mapsto\left.\frac{\partial^2 f}{\partial x^2}\right|_{x=z}\bigr)$
\State $F \gets (x,c \mapsto f(c)+J(c)(x-c))$

\State
\State 
\Comment{Get quantization cell boundaries from $\theta$}
\For{$i \gets 1$ to $n$}
    \State $b_i \gets
        \Call{CellEdges}{\omega_i,\underline{x}_i,\overline{x}_i}$
    \State $m_i \gets |b_i|-1$
\EndFor

\State
\State 
\Comment{Build the quantization cells}
\State $\hat{X} \gets \{1, \dots, m_1\} \times \dots \times \{1, \dots, m_n\}$
\State $\hat{X}_0 \gets \emptyset$
\State $\ell(\hat{x}) \gets
    \left\langle b_{i,\hat{x}_i}\right\rangle_{i=1}^{n}$, $u(\hat{x}) \gets \left\langle b_{i,\hat{x}_i+1}\right\rangle_{i=1}^{n}$
\State $\mathcal{C}(\hat{x}) \gets
    [\ell(\hat{x}),u(\hat{x})]$

\State 
\Comment{Identify successors of every $\hat{x}$}
\ForAll{$\hat{x}\in\hat{X}$}
    \State $\ell \gets \ell(\hat{x})$, $u \gets u(\hat{x})$, $c \gets (\ell + u)/2$

    \State $\mathcal{V} \gets
        \displaystyle \{\ell_1,u_1\} \times \dots \times \{\ell_n,u_n\}$

    \State $\mathcal{Y} \gets
        \{F(v;c)\mid v\in\mathcal{V}\}$

    \State $e \gets
        \Call{TaylorRemainder}{H,[\ell,u]}$

    \State $\ell' \gets
        \left\langle
        \min_{y\in\mathcal{Y}} y_i
        \right\rangle_{i=1}^{n}
        -e$

    \State $u' \gets
        \left\langle
        \max_{y\in\mathcal{Y}} y_i
        \right\rangle_{i=1}^{n}
        +e$

    \State $\hat{f}(\hat{x}) \gets
        \{
        \hat{x}'\in\hat{X}
        \mid
        [\ell',u']
        \cap \mathcal{C}(\hat{x}')
        \neq\emptyset
        \}$
    \State $\hat{X}_0 \gets
        \{
        \hat{x}'\in\hat{X}
        \mid
        [\underline{y},\overline{y}]
        \cap \mathcal{C}(\hat{x}')
        \neq\emptyset
        \}$
\EndFor

\State \Return $\hat{X}, \hat{X}_0, \hat{f}$
\end{algorithmic}
\label{alg:buildabstractmodel}
\end{algorithm}

\section{Additional Experiments}

Our additional experiments address three questions. First, what is the effect of temperature on the correlation between $S^3$ and the reverse simulation metric (and the mean/median $\delta$s)? Second, how do alternative $S^3$ hyperparameters affect grid optimization? Third, do these effects agree with our findings in the proxy validation experiments? The code used to conduct all experiments is publicly available~\cite{tea_s3_abstraction_optimization}.

\subsection{Proxy validation}

\subsubsection{Temperature trade off max vs mean correlation.} In the discussion of Table~\ref{tab:correlation1}, we identified a trade-off between the correlation of $S^3$ with worst-case and abstraction-wide discrepancy statistics. This trend persists across Tables~\ref{tab:correlation1}, \ref{tab:correlation2}, and \ref{tab:correlation3}: increasing the temperature hyperparameters of $S^3$ produces a modest but consistent increase in its correlation with the mean state-wise discrepancy $\overline{\delta}_H$, accompanied by a slight reduction in its correlation with the worst-case reverse simulation metric $\sigma^\leftarrow_H$. This behavior is consistent with the temperature-scaled log-sum-exp aggregations used in $S^3$. At low temperatures, log-sum-exp closely approximates the maximum of its inputs, causing the surrogate to emphasize the largest discrepancies. As the temperature increases, the aggregation becomes less dominated by the maximum and more sensitive to the distribution of discrepancies across abstract states and time steps. The temperature parameters therefore control the extent to which $S^3$ reflects worst-case versus abstraction-wide discrepancy.

\begin{table*}
    \centering
    \caption{Correlation coefficients between $S^3$ and simulation metrics, and their evaluation time on randomly sampled parameters. LSE temperatures are set to $\tau_1=\tau_2=0.5$.}
    \label{tab:correlation2}

    \setlength{\tabcolsep}{5pt}
    \renewcommand{\arraystretch}{1.00}
    \small

    \begin{tabular}{
        c
        c
        cc
        cc
        cc
        cc
    }
        \toprule
        & &
        \multicolumn{2}{c}{$\sigma^\leftarrow_H$} &
        \multicolumn{2}{c}{$\overline{\delta}_H$} &
        \multicolumn{2}{c}{$\delta^{0.50}_H$} &
        \multicolumn{2}{c}{Time per sample (s)} \\

        \cmidrule(lr){3-4}
        \cmidrule(lr){5-6}
        \cmidrule(lr){7-8}
        \cmidrule(lr){9-10}

        & $H$
        & $r_p$
        & $r_s$
        & $r_p$
        & $r_s$
        & $r_p$
        & $r_s$
        & $T$
        & $\hat{T}$ \\
        \midrule

        \multirow{5}{*}{\rotatebox[origin=c]{90}{Spiral}}

        & 1
        & $0.76^{+0.08}_{-0.10}$ & $0.72^{+0.09}_{-0.12}$
        & $0.50^{+0.12}_{-0.14}$ & $0.48^{+0.14}_{-0.16}$
        & $-$ & $-$
        & $5.07^{+0.19}_{-0.21}$ & $0.01^{+0.03}_{-0.01}$ \\
        
        & 2
        & $0.72^{+0.10}_{-0.12}$ & $0.65^{+0.11}_{-0.15}$
        & $0.81^{+0.06}_{-0.09}$ & $0.79^{+0.07}_{-0.10}$
        & $0.70^{+0.08}_{-0.12}$ & $0.68^{+0.10}_{-0.13}$
        & $5.25^{+0.18}_{-0.21}$ & $0.01^{+0.04}_{-0.01}$ \\
        
        & 3
        & $0.70^{+0.10}_{-0.14}$ & $0.67^{+0.11}_{-0.15}$
        & $0.78^{+0.07}_{-0.10}$ & $0.76^{+0.08}_{-0.12}$
        & $0.72^{+0.08}_{-0.12}$ & $0.70^{+0.10}_{-0.13}$
        & $6.05^{+0.21}_{-0.24}$ & $0.01^{+0.04}_{-0.01}$ \\
        
        & 4
        & $0.60^{+0.13}_{-0.18}$ & $0.62^{+0.12}_{-0.17}$
        & $0.74^{+0.07}_{-0.10}$ & $0.74^{+0.08}_{-0.12}$
        & $0.70^{+0.09}_{-0.12}$ & $0.71^{+0.09}_{-0.13}$
        & $6.85^{+0.23}_{-0.27}$ & $0.01^{+0.04}_{-0.01}$ \\
        
        & 5
        & $0.50^{+0.16}_{-0.18}$ & $0.52^{+0.14}_{-0.18}$
        & $0.73^{+0.07}_{-0.10}$ & $0.73^{+0.09}_{-0.12}$
        & $0.70^{+0.09}_{-0.15}$ & $0.72^{+0.09}_{-0.13}$
        & $7.23^{+0.25}_{-0.28}$ & $0.01^{+0.03}_{-0.01}$ \\

        \midrule

        \multirow{5}{*}{\rotatebox[origin=c]{90}{Unicycle}}
        & 1
        & $0.84^{+0.05}_{-0.06}$ & $0.87^{+0.04}_{-0.07}$
        & $0.00^{+0.18}_{-0.18}$ & $0.03^{+0.19}_{-0.19}$
        & $-0.13^{+0.16}_{-0.16}$ & $-0.10^{+0.19}_{-0.18}$
        & $13.21^{+0.19}_{-0.21}$ & $0.01^{+0.04}_{-0.01}$ \\
        
        & 2
        & $0.70^{+0.08}_{-0.12}$ & $0.68^{+0.10}_{-0.14}$
        & $0.62^{+0.11}_{-0.14}$ & $0.59^{+0.012}_{-0.16}$
        & $0.23^{+0.16}_{-0.18}$ & $0.21^{+0.18}_{-0.19}$
        & $15.52^{+0.23}_{-0.25}$ & $0.01^{+0.04}_{-0.01}$ \\
        
        & 3
        & $0.56^{+0.12}_{-0.14}$ & $0.56^{+0.12}_{-0.16}$
        & $0.59^{+0.11}_{-0.14}$ & $0.59^{+0.12}_{-0.16}$
        & $0.50^{+0.14}_{-0.17}$ & $0.48^{+0.14}_{-0.18}$
        & $21.31^{+0.40}_{-0.42}$ & $0.01^{+0.05}_{-0.01}$ \\
        
        & 4
        & $0.46^{+0.12}_{-0.14}$ & $0.45^{+0.13}_{-0.16}$
        & $0.56^{+0.11}_{-0.13}$ & $0.58^{+0.11}_{-0.16}$
        & $0.59^{+0.10}_{-0.12}$ & $0.61^{+0.11}_{-0.15}$
        & $34.52^{+1.01}_{-0.97}$ & $0.01^{+0.06}_{-0.01}$ \\
        
        & 5
        & $0.47^{+0.12}_{-0.15}$ & $0.48^{+0.14}_{-0.16}$
        & $0.53^{+0.11}_{-0.13}$ & $0.56^{+0.12}_{-0.16}$
        & $0.60^{+0.10}_{-0.13}$ & $0.60^{+0.12}_{-0.15}$
        & $64.66^{+2.85}_{-2.53}$ & $0.01^{+0.05}_{-0.01}$ \\

        \midrule

        \multirow{5}{*}{\rotatebox[origin=c]{90}{Mountain car}}
        & 1
        & $0.76^{+0.08}_{-0.10}$ & $0.76^{+0.08}_{-0.12}$
        & $-0.07^{+0.19}_{-0.19}$ & $-0.05^{+0.20}_{-0.20}$
        & $-0.46^{+0.20}_{-0.16}$ & $-0.40^{+0.19}_{-0.17}$
        & $3.64^{+0.16}_{-0.19}$ & $0.02^{+0.10}_{-0.02}$ \\
        
        & 2
        & $0.75^{+0.08}_{-0.10}$ & $0.75^{+0.09}_{-0.12}$
        & $0.85^{+0.05}_{-0.07}$ & $0.83^{+0.06}_{-0.09}$
        & $0.31^{+0.20}_{-0.23}$ & $0.28^{+0.19}_{-0.22}$
        & $5.00^{+0.20}_{-0.24}$ & $0.02^{+0.09}_{-0.02}$ \\
        
        & 3
        & $0.75^{+0.07}_{-0.09}$ & $0.77^{+0.08}_{-0.11}$
        & $0.89^{+0.04}_{-0.06}$ & $0.87^{+0.05}_{-0.07}$
        & $0.64^{+0.10}_{-0.14}$ & $0.61^{+0.12}_{-0.17}$
        & $6.41^{+0.23}_{-0.26}$ & $0.02^{+0.08}_{-0.02}$ \\
        
        & 4
        & $0.78^{+0.06}_{-0.08}$ & $0.77^{+0.08}_{-0.11}$
        & $0.88^{+0.04}_{-0.05}$ & $0.88^{+0.04}_{-0.06}$
        & $0.74^{+0.08}_{-0.12}$ & $0.70^{+0.10}_{-0.14}$
        & $7.77^{+0.26}_{-0.30}$ & $0.02^{+0.09}_{-0.02}$ \\
        
        & 5
        & $0.76^{+0.07}_{-0.10}$ & $0.78^{+0.07}_{-0.10}$
        & $0.87^{+0.04}_{-0.06}$ & $0.86^{+0.05}_{-0.07}$
        & $0.75^{+0.08}_{-0.10}$ & $0.74^{+0.09}_{-0.13}$
        & $9.08^{+0.31}_{-0.36}$ & $0.02^{+0.07}_{-0.02}$ \\

        \bottomrule
    \end{tabular}
\end{table*}

\begin{table*}
    \centering
    \caption{Correlation coefficients between $S^3$ and simulation metrics, and their evaluation time on randomly sampled parameters. LSE temperatures are set to $\tau_1=\tau_2=1.0$.}
    \label{tab:correlation3}

    \setlength{\tabcolsep}{5pt}
    \renewcommand{\arraystretch}{1.00}
    \small

    \begin{tabular}{
        c
        c
        cc
        cc
        cc
        cc
    }
        \toprule
        & &
        \multicolumn{2}{c}{$\sigma^\leftarrow_H$} &
        \multicolumn{2}{c}{$\overline{\delta}_H$} &
        \multicolumn{2}{c}{$\delta^{0.50}_H$} &
        \multicolumn{2}{c}{Time per sample (s)} \\

        \cmidrule(lr){3-4}
        \cmidrule(lr){5-6}
        \cmidrule(lr){7-8}
        \cmidrule(lr){9-10}

        & $H$
        & $r_p$
        & $r_s$
        & $r_p$
        & $r_s$
        & $r_p$
        & $r_s$
        & $T$
        & $\hat{T}$ \\
        \midrule

        \multirow{5}{*}{\rotatebox[origin=c]{90}{Spiral}}
        & 1
        & $0.70^{+0.09}_{-0.12}$ & $0.67^{+0.10}_{-0.13}$
        & $0.48^{+0.12}_{-0.15}$ & $0.47^{+0.14}_{-0.16}$
        & $-$ & $-$
        & $2.91^{+0.05}_{-0.04}$ & $0.01^{+0.04}_{-0.01}$ \\
        
        & 2
        & $0.66^{+0.10}_{-0.13}$ & $0.61^{+0.13}_{-0.17}$
        & $0.81^{+0.06}_{-0.09}$ & $0.78^{+0.07}_{-0.10}$
        & $0.73^{+0.08}_{-0.11}$ & $0.71^{+0.09}_{-0.12}$
        & $3.06^{+0.04}_{-0.04}$ & $0.01^{+0.03}_{-0.01}$ \\
        
        & 3
        & $0.65^{+0.11}_{-0.14}$ & $0.63^{+0.12}_{-0.16}$
        & $0.78^{+0.07}_{-0.10}$ & $0.76^{+0.08}_{-0.11}$
        & $0.73^{+0.08}_{-0.11}$ & $0.72^{+0.09}_{-0.12}$
        & $3.54^{+0.06}_{-0.05}$ & $0.01^{+0.04}_{-0.01}$ \\
        
        & 4
        & $0.56^{+0.13}_{-0.19}$ & $0.59^{+0.13}_{-0.18}$
        & $0.74^{+0.07}_{-0.10}$ & $0.73^{+0.08}_{-0.12}$
        & $0.72^{+0.08}_{-0.11}$ & $0.72^{+0.09}_{-0.12}$
        & $4.04^{+0.12}_{-0.06}$ & $0.01^{+0.03}_{-0.01}$ \\
        
        & 5
        & $0.47^{+0.15}_{-0.19}$ & $0.50^{+0.14}_{-0.19}$
        & $0.73^{+0.07}_{-0.10}$ & $0.72^{+0.09}_{-0.12}$
        & $0.72^{+0.08}_{-0.12}$ & $0.72^{+0.09}_{-0.13}$
        & $4.23^{+0.10}_{-0.07}$ & $0.01^{+0.03}_{-0.01}$ \\

        \midrule

        \multirow{5}{*}{\rotatebox[origin=c]{90}{Unicycle}}
        & 1
        & $0.83^{+0.04}_{-0.07}$ & $0.86^{+0.05}_{-0.06}$
        & $-0.01^{+0.19}_{-0.17}$ & $0.00^{+0.19}_{-0.19}$
        & $-0.11^{+0.15}_{-0.15}$ & $-0.12^{+0.19}_{-0.18}$
        & $23.27^{+0.57}_{-0.84}$ & $0.01^{+0.05}_{-0.01}$ \\
        
        & 2
        & $0.75^{+0.08}_{-0.11}$ & $0.71^{+0.10}_{-0.14}$
        & $0.70^{+0.09}_{-0.11}$ & $0.69^{+0.10}_{-0.12}$
        & $0.28^{+0.15}_{-0.17}$ & $0.29^{+0.16}_{-0.18}$
        & $27.19^{+0.66}_{-1.00}$ & $0.01^{+0.06}_{-0.01}$ \\
        
        & 3
        & $0.54^{+0.12}_{-0.15}$ & $0.56^{+0.13}_{-0.16}$
        & $0.59^{+0.11}_{-0.14}$ & $0.60^{+0.12}_{-0.16}$
        & $0.51^{+0.13}_{-0.16}$ & $0.49^{+0.14}_{-0.17}$
        & $36.80^{+0.96}_{-1.29}$ & $0.01^{+0.05}_{-0.01}$ \\
        
        & 4
        & $0.46^{+0.12}_{-0.15}$ & $0.45^{+0.14}_{-0.17}$
        & $0.57^{+0.11}_{-0.13}$ & $0.57^{+0.12}_{-0.16}$
        & $0.59^{+0.10}_{-0.12}$ & $0.60^{+0.11}_{-0.15}$
        & $58.44^{+2.00}_{-2.19}$ & $0.01^{+0.06}_{-0.01}$ \\
        
        & 5
        & $0.48^{+0.12}_{-0.15}$ & $0.49^{+0.14}_{-0.17}$
        & $0.56^{+0.10}_{-0.13}$ & $0.58^{+0.12}_{-0.16}$
        & $0.62^{+0.10}_{-0.13}$ & $0.61^{+0.12}_{-0.15}$
        & $105.40^{+4.80}_{-4.58}$ & $0.01^{+0.06}_{-0.01}$ \\

        \midrule

        \multirow{5}{*}{\rotatebox[origin=c]{90}{Mountain car}}
        & 1
        & $0.75^{+0.08}_{-0.10}$ & $0.75^{+0.08}_{-0.12}$
        & $-0.06^{+0.19}_{-0.19}$ & $-0.05^{+0.20}_{-0.20}$
        & $-0.47^{+0.20}_{-0.16}$ & $-0.41^{+0.19}_{-0.17}$
        & $3.43^{+0.18}_{-0.18}$ & $0.02^{+0.10}_{-0.02}$ \\

        & 2
        & $0.73^{+0.09}_{-0.11}$ & $0.74^{+0.09}_{-0.13}$
        & $0.86^{+0.05}_{-0.07}$ & $0.84^{+0.06}_{-0.09}$
        & $0.33^{+0.19}_{-0.22}$ & $0.29^{+0.19}_{-0.22}$
        & $4.73^{+0.22}_{-0.23}$ & $0.02^{+0.11}_{-0.02}$ \\

        & 3
        & $0.74^{+0.07}_{-0.10}$ & $0.75^{+0.08}_{-0.12}$
        & $0.90^{+0.04}_{-0.06}$ & $0.88^{+0.05}_{-0.06}$
        & $0.66^{+0.10}_{-0.14}$ & $0.62^{+0.12}_{-0.16}$
        & $6.15^{+0.23}_{-0.25}$ & $0.03^{+0.13}_{-0.03}$ \\

        & 4
        & $0.76^{+0.06}_{-0.09}$ & $0.75^{+0.08}_{-0.12}$
        & $0.89^{+0.04}_{-0.05}$ & $0.88^{+0.04}_{-0.06}$
        & $0.75^{+0.08}_{-0.11}$ & $0.72^{+0.09}_{-0.14}$
        & $7.47^{+0.27}_{-0.28}$ & $0.03^{+0.12}_{-0.03}$ \\

        & 5
        & $0.74^{+0.08}_{-0.10}$ & $0.76^{+0.08}_{-0.10}$
        & $0.88^{+0.04}_{-0.05}$ & $0.87^{+0.05}_{-0.07}$
        & $0.77^{+0.07}_{-0.10}$ & $0.75^{+0.08}_{-0.12}$
        & $8.90^{+0.30}_{-0.31}$ & $0.02^{+0.12}_{-0.02}$ \\

        \bottomrule
    \end{tabular}
\end{table*}

\subsubsection{Hyperparameters.}

We set each component of $\tilde{e}$ to one half of the corresponding average cell width. This is a practical compromise between the sufficient containment condition in Proposition~\ref{prop:etilde} and the risk of under-approximating the concretized abstract reachable sets. Enforcing the certified condition $\tilde{e}_i \geq \max_j \eta_{i,j}$ is also inconvenient during optimization because the cell widths $\eta_{i,j}$ depend on the quantization parameters $\theta$. Consequently, defining $\tilde{e}$ using the maximum cell width --- or using a width quantile as a less conservative heuristic --- would make $\tilde{e}$ parameter-dependent and introduce additional nonsmooth operations into the objective. (Note that such heuristics do not make the abstraction unsound.)

By contrast, the average cell width in dimension $i$ is determined solely by the domain extent and grid resolution,
\begin{equation}
    \overline{\eta}_i = \frac{\overline{x}_i-\underline{x}_i}{m_i},
\end{equation}
and is therefore independent of $\theta$. This yields a fixed inflation vector that can be selected before optimization and differentiated through without additional dependence on the learned quantization parameters.

We held the grid resolution fixed throughout the hyperparameter experiments. For Spiral and Mountain Car, we used $m_i=50$ in each dimension, whereas for Unicycle we used $m_i=20$. These resolutions were selected primarily to control the computational cost of evaluating the finite-horizon simulation metric across the number of parameter samples required to obtain sufficiently tight bootstrap confidence intervals (see evaluation times in Tables~\ref{tab:correlation1}, \ref{tab:correlation2}, and \ref{tab:correlation3}).

\subsection{Optimization performance}

\begin{figure*}[t]
    \centering
    \includegraphics[width=0.9\linewidth]{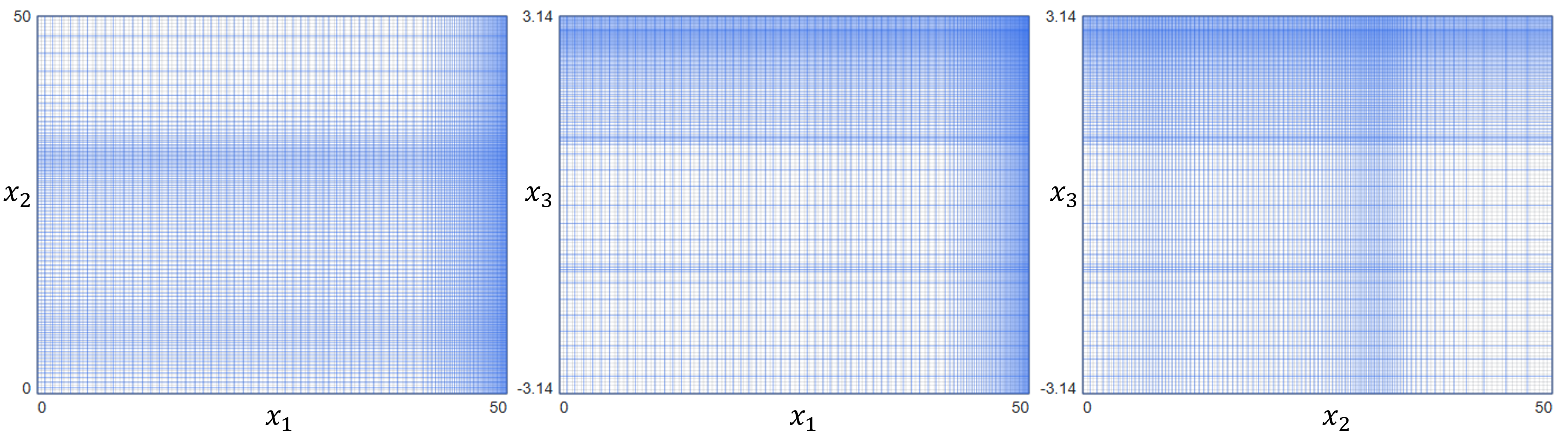}
    \caption{All orthogonal projections of the final grid after optimization of Unicycle ($m_i=100$, $H=5, \tau_1=\tau_2=0.1$).}
\label{fig:uni-grid}
\end{figure*}

\begin{figure}
    \centering
    \includegraphics[width=0.8\linewidth]{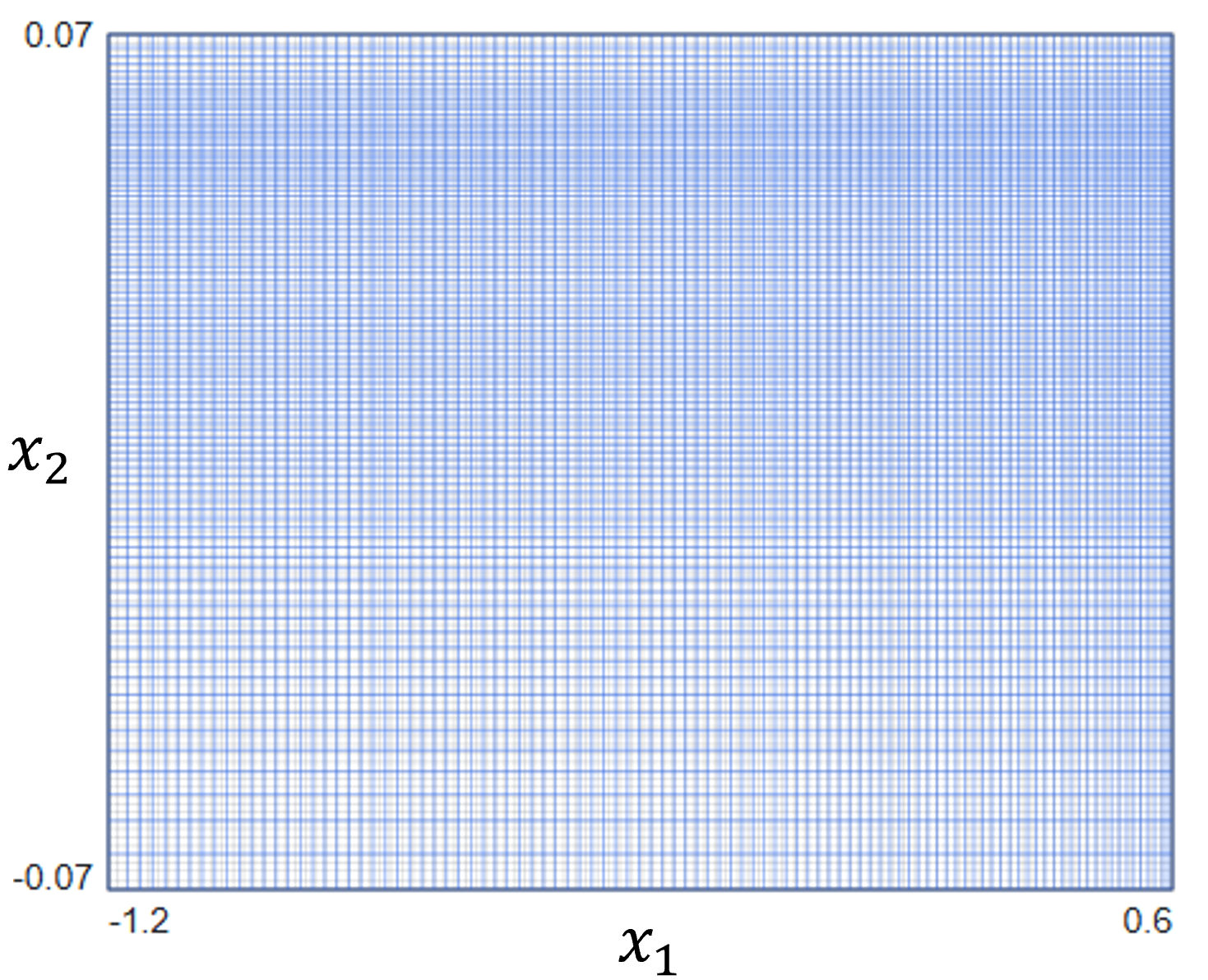}
    \caption{Final grid after optimization of mountain car ($m_i=100$, $H=1, \tau=0.1$).}
\label{fig:mc-grid}
\end{figure}

\paragraph{Effect of alternative $S^3$ hyperparameters on grid optimization.} Table~\ref{tab:more-optiimzation} shows that the effect of the $S^3$ hyperparameters is relatively system-dependent. Spiral (the simplest system) is largely insensitive to temperature because the solution --- which is the uniform grid --- is easy to find. For a fixed horizon, $\tau=0.1$ and $\tau=1.0$ produce nearly identical changes in the simulation metrics and verification recall. Mountain Car exhibits the most sensitivity. At $H=1$, optimization generally produces modest improvements in the worst-case, mean, and median discrepancies. At longer horizons, particularly for $\tau=1.0$, optimization tends to produce larger reductions in the median and, in several settings, the mean discrepancy, but these improvements are accompanied by degradation in the worst-case metric $\sigma^\leftarrow_H$ and reduced verification recall. Unicycle exhibits more variable behavior. At resolution $50^3$, intermediate and long horizons can substantially reduce the worst-case metric, with the largest reduction obtained at $H=5$ and $\tau=0.1$. Note that the simulation metrics in Table~\ref{tab:more-optiimzation} should not be compared row-wise over $H\in\{1, 3, 5\}$ within a case study and resolution because the simulation metric grows monotonically in $H$. Rather, they should be compared among $\tau \in \{0.1, 1.0\}$.

\paragraph{Agreement with the proxy-validation experiments.} These optimization results qualitatively agree with the proxy-validation experiments reported through Tables~\ref{tab:correlation1}, \ref{tab:correlation2}, and \ref{tab:correlation3}. The correlation analysis showed that increasing the LSE temperature increases the association of $S^3$ with abstraction-wide statistics while reducing its association with the worst-case metric, and that correlation with $\sigma^\leftarrow_H$ generally weakens as the horizon increases. The Mountain Car results provide the strongest confirmation of this trade-off. At $H\in\{3,5\}$, $\tau=1.0$ generally yields greater median reductions than $\tau=0.1$, but also produces a substantially larger increase in $\sigma^\leftarrow_H$ and a smaller, or negative, change in recall. Similarly, for Unicycle at resolution $50^3$ and $H=5$, the lower-temperature objective produces a considerably larger reduction in the worst-case metric than the higher-temperature objective. Spiral shows little separation between temperatures, suggesting that its optimization landscape is insensitive to this trade-off.

\paragraph{Optimizer settings.} For Spiral and Mountain Car, the optimizer settings were held fixed across all evaluated horizons and temperatures. Spiral was optimized for $1000$ gradient steps using a learning rate of $200$. Mountain Car was optimized for $500$ steps using a learning rate of $320$. Unicycle required case-by-case tuning. For models at resolution $50^3$, all configurations were optimized for $800$ steps using a learning rate of $0.005$. At resolution $100^3$ and $H\in{1,3}$ configurations, $800$ steps and a learning rate of $10.0$ were used. For the $H=5$ configurations, $800$ steps and a learning rate of $5\times10^{-5}$.

\paragraph{Grid visualizations.} Figures~\ref{fig:uni-grid} and \ref{fig:mc-grid} show final, optimized grids of Unicycle and Mountain Car, respectively. Spiral was omitted, as the optimal grid is simply uniform. We observe the most substantial variation in gap widths in the Unicycle case. Here, the optimizer learns to allocate more granularity to higher turn rates ($x_3$), which occur at the boundaries of the obstacle or goal. While the variation in gap width is more mild in the Mountain Car case, we observe that more granularity is allocated to higher velocities (where the reachable set expands rapidly).

\begin{table*}
    \centering
    \caption{$S^3$ + SGD optimization results for multiple horizons $H$ and temperatures $\tau_1=\tau_2=\tau$.}
    \setlength{\tabcolsep}{4.50pt}
    \renewcommand{\arraystretch}{0.8}
    \footnotesize

    \begin{tabular}{
        c
        c
        c
        cc
        cc
        cc
        cc
    }
        \toprule
        & & &
        \multicolumn{6}{c}{Sim. Metrics} &
        \multicolumn{2}{c}{Task} \\

        \cmidrule(lr){4-9}
        \cmidrule(lr){10-11}

        & $m_i$
        & $(H,\tau)$
        & $\sigma^\leftarrow_H \downarrow$
        & $\Delta \sigma^\leftarrow_H \downarrow$
        & $\overline{\delta}_H \downarrow$
        & $\Delta \overline{\delta}_H \downarrow$
        & $\delta_H^{0.5} \downarrow$
        & $\Delta \delta_H^{0.5} \downarrow$
        & Recall $\uparrow$
        & $\Delta \text{Recall} \uparrow$ \\
        \midrule

        \multirow[c]{12}{*}{\rotatebox[origin=c]{90}{Spiral}}
        & \multirow[c]{6}{*}{70}
        & $(1,0.1)$ & $0.14$ & $\shortneg 0.021$ & $0.027$ & $\shortneg 0.00067$ & $0.0$ & $0.0$ & $0.87$ & $0.011$ \\
        && $(1,1.0)$ & $0.14$ & $\shortneg 0.020$ & $0.027$ & $\shortneg 0.00067$ & $0.0$ & $0.0$ & $0.87$ & $0.011$ \\
        && $(3,0.1)$ & $2.1$ & $\shortneg 0.024$ & $0.37$ & $\shortneg 0.0045$ & $0.43$ & $0.013$ & $0.87$ & $0.011$ \\
        && $(3,1.0)$ & $2.1$ & $\shortneg 0.023$ & $0.37$ & $\shortneg 0.0043$ & $0.43$ & $0.013$ & $0.87$ & $0.011$ \\
        && $(5,0.1)$ & $3.9$ & $\shortneg 0.029$ & $0.67$ & $\shortneg 0.0065$ & $0.71$ & $0.012$ & $0.87$ & $0.011$ \\
        && $(5,1.0)$ & $3.9$ & $\shortneg 0.029$ & $0.67$ & $\shortneg 0.0065$ & $0.71$ & $0.012$ & $0.87$ & $0.011$ \\

        \cmidrule(lr){2-11}

        & \multirow[c]{6}{*}{100}
        & $(1,0.1)$ & $0.10$ & $\shortneg 0.019$ & $0.019$ & $\shortneg 0.00040$ & $0.0$ & $0.0$ & $0.92$ & $0.0058$ \\
        && $(1,1.0)$ & $0.10$ & $\shortneg 0.019$ & $0.019$ & $\shortneg 0.00034$ & $0.0$ & $0.0$ & $0.93$ & $0.0078$ \\
        && $(3,0.1)$ & $2.1$ & $\shortneg 0.016$ & $0.27$ & $\shortneg 0.0019$ & $0.30$ & $0.0061$ & $0.93$ & $0.014$ \\
        && $(3,1.0)$ & $2.1$ & $\shortneg 0.016$ & $0.27$ & $\shortneg 0.0016$ & $0.30$ & $0.0061$ & $0.93$ & $0.0082$ \\
        && $(5,0.1)$ & $4.0$ & $\shortneg 0.0074$ & $0.50$ & $\shortneg 0.0033$ & $0.50$ & $0.0019$ & $0.93$ & $0.012$ \\
        && $(5,1.0)$ & $4.0$ & $\shortneg 0.0074$ & $0.50$ & $\shortneg 0.0045$ & $0.50$ & $0.0019$ & $0.93$ & $0.0086$ \\

        \midrule

        \multirow[c]{12}{*}{\rotatebox[origin=c]{90}{Unicycle}}
        & \multirow[c]{6}{*}{50}
        & $(1,0.1)$ & $3.0$ & $0.0048$ & $0.24$ & $\shortneg 0.00076$ & $0.090$ & $\shortneg 0.00026$ & $0.24$ & $\shortneg 0.0053$ \\
        && $(1,1.0)$ & $3.0$ & $0.0014$ & $0.24$ & $\shortneg 3.7\text{e} \shortneg5$ & $0.091$ & $\shortneg 5.1\text{e} \shortneg5$ & $0.24$ & $\shortneg 3.5\text{e} \shortneg5$ \\
        && $(3,0.1)$ & $9.9$ & $\shortneg 0.15$ & $1.9$ & $0.0023$ & $2.0$ & $0.00038$ & $0.24$ & $\shortneg 0.00092$ \\
        && $(3,1.0)$ & $9.8$ & $\shortneg 0.16$ & $1.9$ & $\shortneg 0.00076$ & $2.0$ & $\shortneg 0.00087$ & $0.24$ & $\shortneg 0.0042$ \\
        && $(5,0.1)$ & $19$ & $\shortneg 1.6$ & $4.2$ & $0.0085$ & $4.0$ & $0.0068$ & $0.24$ & $\shortneg 0.0048$ \\
        && $(5,1.0)$ & $20$ & $\shortneg 0.37$ & $4.2$ & $0.025$ & $4.0$ & $0.017$ & $0.22$ & $\shortneg 0.027$ \\

        \cmidrule(lr){2-11}

        & \multirow[c]{6}{*}{100}
        & $(1,0.1)$ & $2.8$ & $1.4$ & $0.12$ & $0.0034$ & $0.038$ & $\shortneg 0.0022$ & $0.70$ & $\shortneg 0.098$ \\
        && $(1,1.0)$ & $3.1$ & $1.6$ & $0.12$ & $0.0028$ & $0.037$ & $\shortneg 0.0028$ & $0.67$ & $\shortneg 0.13$ \\
        && $(3,0.1)$ & $9.5$ & $1.6$ & $0.97$ & $0.041$ & $0.86$ & $\shortneg 0.12$ & $0.75$ & $\shortneg 0.051$ \\
        && $(3,1.0)$ & $9.9$ & $2.0$ & $0.93$ & $0.00096$ & $0.98$ & $\shortneg 0.00079$ & $0.79$ & $\shortneg 0.015$ \\
        && $(5,0.1)$ & $18$ & $0.00045$ & $2.1$ & $\shortneg 3.8\text{e-}5$ & $1.9$ & $0.00015$ & $0.80$ & $0.00024$ \\
        && $(5,1.0)$ & $18$ & $\shortneg 5.8\text{e-}5$ & $2.1$ & $\shortneg 0.00021$ & $1.9$ & $\shortneg 0.00019$ & $0.80$ & $6.8\text{e-}6$ \\

        \midrule

        \multirow[c]{12}{*}{\rotatebox[origin=c]{90}{Mountain Car}}
        & \multirow[c]{6}{*}{70}
        &  $(1,0.1)$ & $0.065$ & $\shortneg 0.00090$ & $0.0019$ & $\shortneg0.00046$ & $0.00038$ & $\shortneg 8.6 \text{e} \shortneg5$ & $0.21$ & $0.082$ \\
        && $(1,1.0)$ & $0.061$ & $\shortneg 0.0042$ & $0.0019$ & $\shortneg0.00042$ & $0.00034$ & $\shortneg0.00013$ & $0.26$ & $0.14$ \\
        && $(3,0.1)$ & $0.085$ & $0.013$ & $0.031$ & $\shortneg0.0034$ & $0.029$ &$ \shortneg0.00052$ & $0.20$ & $0.072$ \\
        && $(3,1.0)$ & $0.16$ & $0.092$ & $0.033$ & $\shortneg 0.0011$ & $0.015$ & $\shortneg0.014$ & $0.16$ & $0.033$ \\
        && $(5,0.1)$ & $0.18$ & $0.033$ & $0.060$ & $\shortneg 0.0090$ & $0.056$ & $\shortneg0.010$ & $0.21$ & $0.089$ \\
        && $(5,1.0)$ & $0.33$ & $0.19$ & $0.069$ & $\shortneg 3.2 \text{e} \shortneg 5$ & $0.011$ & $\shortneg0.055$ & $0.14$ & $0.018$ \\

        \cmidrule(lr){2-11}

        & \multirow[c]{6}{*}{100}
        & $(1,0.1)$ & $0.064$ & $\shortneg 0.0012$ & $0.0013$ & $\shortneg 0.00036$ & $0.00026$ & $\shortneg 3.9 \text{e} \shortneg 5$ & $0.35$ & $0.083$ \\
        && $(1,1.0)$ & $0.064$ & $\shortneg 0.0015$ & $0.0013$ & $\shortneg 0.00033$ & $0.00026$ & $\shortneg 4.0\text{e} \shortneg5$ & $0.35$ & $0.081$ \\
        && $(3,0.1)$ & $0.082$ & $0.015$ & $0.025$ & $0.00086$ & $0.024$ & $0.0031$ & $0.26$ & $\shortneg 0.0045$ \\
        && $(3,1.0)$ & $0.15$ & $0.087$ & $0.021$ & $\shortneg 0.0032$ & $0.0030$ & $\shortneg 0.018$ & $0.23$ & $\shortneg 0.033$ \\
        && $(5,0.1)$ & $0.18$ & $0.067$ & $0.045$ & $\shortneg 0.0034$ & $0.047$ & $0.00038$ & $0.27$ & $0.0035$ \\
        && $(5,1.0)$ & $0.26$ & $0.14$ & $0.041$ & $\shortneg 0.0067$ & $0.0031$ & $\shortneg 0.044$ & $0.22$ & $\shortneg 0.047$ \\

        \bottomrule
    \end{tabular}
    \label{tab:more-optiimzation}
\end{table*}

\end{document}